\documentclass[letterpaper, 11pt]{article}
\usepackage{graphicx}
\usepackage{amsmath}
\usepackage{amsfonts}
\usepackage{amsthm}
 \usepackage{amscd}
\usepackage{mathrsfs}
\usepackage{caption}
\usepackage{subcaption}
\usepackage{float}
\usepackage[top=1in,left=1in,right=1in,bottom=1in]{geometry}
\usepackage{hyperref}
\hypersetup{
	colorlinks = true,
	linkcolor  = red,
	citecolor  = green,
	filecolor  = cyan,
	urlcolor   = magenta,}

\newtheorem{theorem}{Theorem}[section]

\newtheorem{corollary}[theorem]{Corollary}

\newtheorem{example}[theorem]{Example}

\numberwithin{equation}{section}

\numberwithin{equation}{section}

\title{Inverse scattering  for the linear system associated with the coupled Gerdjikov--Ivanov equations}
\author{Ramazan Ercan\\
Department of Mathematics\\
California State University San Marcos\\
San Marcos, CA 92096, USA
}

\date{}

\begin{document}

\maketitle

\begin{abstract}
We consider a certain first-order linear system of ordinary differential equations, and we analyze the direct and inverse scattering problems for that linear system.
The linear system involves two potentials in the Schwartz class, and those potentials linearly depend on the spectral parameter.
This linear system is related to the integrable system of nonlinear partial differential equations known as the DNLS (derivative nonlinear Schr\"odinger) system III, which is also known as
the Gerdjikov--Ivanov system.
When analyzing the direct problem, we describe the pertinent properties of the Jost solutions and the scattering coefficients.
The bound states poles and the associated normalization constants are represented via a matrix triplet pair, enabling us to deal with
any number of bound states and any multiplicities. The inverse scattering problem comprises the determination of
the two potentials when the reflection coefficients and the bound-state information are available. To solve the inverse problem, we establish a linear system of integral equations where the kernel and
nonhomogeneous term are determined by the Fourier transforms of the reflection coefficients and the matrix triplet pair representing the bound-state information. This system of linear integral equations is the counterpart of the system of Marchenko integral equations available for the AKNS system associated with the integrable NLS (nonlinear Schr\"odinger) system. We recover the potentials from the solution of our established Marchenko integral system.
When we use the time-evolved reflection coefficients and the time-evolved matrix triplets, the corresponding time-evolved
potential pair yields a solution of the Gerdjikov--Ivanov system.

\end{abstract}



\newpage

\section{Introduction}
\label{section1}

We consider the  linear system
\begin{equation}\label{1.1}
\displaystyle\frac{d}{dx}\displaystyle\begin{bmatrix}
\alpha\\
\noalign{\medskip}
\beta
\end{bmatrix}=
\begin{bmatrix}
-i\zeta^2- \displaystyle\frac{i}{2} \,q(x)\,r(x) & \zeta q(x)\\
\noalign{\medskip}
\zeta r(x) & i\zeta^2+\displaystyle\frac{i}{2} \,q(x)\,r(x)
\end{bmatrix}
\begin{bmatrix}
\alpha\\
\noalign{\medskip}
\beta
\end{bmatrix},\qquad x\in\mathbb R,
\end{equation}
where $x$ is the independent variable taking values on the real line
$\mathbb R,$ the complex-valued scalar $\zeta$ is the spectral parameter, 
the coefficients $q(x)$ and $r(x)$ are  complex-valued potentials, 
the scalar quantities $\alpha$ and $\beta$ are the components of the wavefunction $\begin{bmatrix}
\alpha\\
\beta
\end{bmatrix}$ depending on $x$ and $\zeta.$ Since the potentials $q$ and $r$ appear in the off-diagonal
entries of the coefficient matrix in
\eqref{1.1} as
$\zeta q(x)$ and $\zeta r(x),$ we refer to \eqref{1.1} as a linear system with
energy-dependent potentials. This is because the spectral
parameter $\zeta$ in \eqref{1.1} is related to
energy in physical applications.

We assume that the potentials $q$ and $r$ in \eqref{1.1} belong to the Schwartz 
class in $x\in\mathbb R.$ We recall that the Schwartz class $\mathcal S(\mathbb R)$ consists of 
functions of $x$ where the derivatives of all orders exist and are continuous and those derivatives vanish as $x\to\pm\infty$ 
faster than any negative power of $|x|.$ Although our results hold under weaker conditions on the potentials, 
we present our results in the simplest form by assuming that $q$ and $r$ belong to $\mathcal S(\mathbb R).$

If the potentials $q$ and $r$ in
\eqref{1.1} also contain the temporal parameter $t,$ then the wavefunction
components $\alpha$ and $\beta$ also depend on the parameter $t.$
Let us consider the special case where that time dependence is governed by the $2\times 2$ AKNS pair
\cite{AC1991,AKNS1974,APT2003,AS1981,A2009} matrices
$\mathbf X$ and $\mathbf T,$ where we have
\begin{equation}\label{1.2}
\mathbf X:=\begin{bmatrix}
-i\zeta^2- \displaystyle\frac{i}{2} \,qr & \zeta q\\
\noalign{\medskip}
\zeta r & i\zeta^2+\displaystyle\frac{i}{2} \,qr
\end{bmatrix},
\end{equation}
\begin{equation}\label{1.3}
\mathbf T:=\begin{bmatrix}
-2i\zeta^4- i\zeta^2 qr+\displaystyle\frac{1}{2} (q_x  r-q r_x )+\displaystyle\frac{i}{4} q^2 r^2
 & 2\zeta^3 q+i\zeta q_x\\
\noalign{\medskip}
2\zeta^3 r -i\zeta r_x&2 i\zeta^4+i \zeta^2 qr+\displaystyle\frac{1}{2} (qr_x-q_xr)-\displaystyle\frac{i}{4} q^2 r^2
\end{bmatrix}.
\end{equation}
We remark that $\mathbf X$ appearing in \eqref{1.2} is the same as the coefficient matrix in \eqref{1.1} and that
the subscripts in \eqref{1.3} denote the partial derivatives. When the AKNS pair matrices satisfy the matrix equality
\begin{equation*}
\mathbf X_t-\mathbf T_x+\mathbf X\mathbf T-\mathbf T\mathbf X=0,
\end{equation*}
the time-evolved potentials $q(x,t)$ and $r(x,t)$ satisfy
the second-order coupled system of integrable evolution equations given by
\begin{equation}
\label{1.5}
\begin{cases}
i q_t+q_{xx}+iq^2 r_x+\displaystyle\frac{1}{2}q^3 r^2=0,
\\
\noalign{\medskip}
i r_t- r_{xx}+i q_x r^2-\displaystyle\frac{1}{2}q^2 r^3
=0.
\end{cases}
\end{equation}
The nonlinear system \eqref{1.5} is known \cite{GI1983}
as the the Gerdjikov--Ivanov system
or as the DNLS (derivative nonlinear Schr\"odinger) III system.
Along with the DNLS I and DNLS II systems, it has important physical applications
 in propagation of electromagnetic waves in nonlinear media, 
propagation of hydromagnetic waves traveling in a magnetic field, and transmission of ultra short nonlinear pulses in optical fibers. 
The DNLS I system, also known as the Kaup--Newell system, is the integrable system of
nonlinear partial differential equations given by
\begin{equation}
\label{1.6}
\begin{cases}
i \tilde q_t+\tilde q_{xx}-i (\tilde q^2 \, \tilde r)_x=0, \\
i \tilde r_t- \tilde r_{xx}-i (\tilde q \tilde r^2)_x =0,
\end{cases}
\end{equation}
and it is associated with the linear system given by
\begin{equation}\label{1.7}
\displaystyle\frac{d}{dx}\begin{bmatrix}
\tilde\alpha\\
\noalign{\medskip}
\tilde\beta
\end{bmatrix}=
\begin{bmatrix}
-i\zeta^2& \zeta \tilde q(x)
\\
\noalign{\medskip}
\zeta \tilde r(x)&i\zeta^2
\end{bmatrix}
\begin{bmatrix}
\tilde\alpha\\
\noalign{\medskip}
\tilde\beta
\end{bmatrix},\qquad x\in\mathbb R.
\end{equation}
The DNLS II system, also known as the Chen--Lee--Liu system \cite{CLL1979}, is the integrable system of
nonlinear partial differential equations given by
\begin{equation}
\label{1.8}
\begin{cases}
i \hat q_t+\hat q_{xx}-i \hat q \hat q_x \hat r=0, \\
i \hat r_t- \hat r_{xx}-i \hat q \hat r \hat r_x =0,
\end{cases}
\end{equation}
and it is associated with the linear system given by
\begin{equation}\label{1.9}
\displaystyle\frac{d}{dx}\displaystyle\begin{bmatrix}
\hat\alpha\\
\noalign{\medskip}
\hat\beta
\end{bmatrix}=
\begin{bmatrix}
-i\zeta^2& \zeta\, \hat q(x)\\
\noalign{\medskip}
\zeta\, \hat r(x) & i\zeta^2+\displaystyle\frac{i}{2} \,\hat q(x)\,\hat r(x)
\end{bmatrix}
\begin{bmatrix}
\hat\alpha\\
\noalign{\medskip}
\hat\beta
\end{bmatrix},\qquad x\in\mathbb R.
\end{equation}
For the DNLS I and DNLS II systems, we refer the reader to \cite{APT2003,AEU2023a,AEU2023b,OS1998a,OS1998b} and the references therein.

Our aim is to analyze the direct and inverse scattering problems for \eqref{1.1}
and to solve the inverse problem by introducing the Marchenko method for \eqref{1.1}.
In the Marchenko method the potentials are obtained from the solution of a system of linear integral equations
whose kernel and nonhomogeneous terms are determined by the scattering data set.
We refer to that system of linear integral equation as the Marchenko system.
The direct scattering problem for \eqref{1.1} consists of the determination of the scattering data set
when $(q,r)$ is known.
The scattering data set comprises the scattering coefficients and the bound-state information.
The bound-state information is made up of the values of the spectral parameter at which
\eqref{1.1} has square-integrable solutions in $x\in\mathbb R$ 
and the bound-state normalization constants specified at each multiplicity of the bound states. The multiplicity of a bound state
corresponds to the number of linearly independent square-integrable solutions to \eqref{1.1} at that bound-state value of the spectral parameter.
The inverse scattering problem for \eqref{1.1} consists of the determination of the potentials $q$ and $r$
from the scattering data set. In this paper, we bring a solution of the inverse problem for \eqref{1.1} by 
the Marchenko method. In other words, we establish the Marchenko system of linear integral equations for \eqref{1.1}, use the
scattering data set as input to the Marchenko system, and recover the potentials from the solution of the
Marchenko system.

The Marchenko method was first used by Vladimir A. Marchenko himself \cite{AM1963,M1955} for the half-line Schr\"odinger equation and later by
Faddeev \cite{F1967} for the full-line Schr\"odinger equation. Next, it was applied \cite{AKNS1974} to
the linear system \eqref{1.11} of differential equations known as the AKNS system.
After that, it was generalized to various other differential and difference equations and systems of such equations. The development of the Marchenko method related to the DNLS systems
is more challenging because of the dependence
on the spectral parameter of the potential pairs in the corresponding linear systems.
We refer the reader to  \cite{AE2019,AEU2023a,AEU2023b,E2018}
for the Marchenko method for \eqref{1.7} and to \cite{U2025} for the Marchenko method
for \eqref{1.9}. In those Marchenko methods, $(\tilde q,\tilde r)$ in \eqref{1.7} and
$(\hat q,\hat r)$ in \eqref{1.9} are explicitly obtained from
the solution of the respective Marchenko systems. In an analogous manner, the 
solutions to the integrable nonlinear systems \eqref{1.6} and \eqref{1.8} are recovered
from the solutions to the respective Marchenko systems when the time-evolved
scattering data sets are used as input.
The establishment of the Marchenko method for \eqref{1.1} presented here follows the ideas used to derive the Marchenko method \cite{AE2019,AEU2023a,AEU2023b,E2018} for the linear system \eqref{1.7} and the Marchenko method \cite{U2025} for the linear system \eqref{1.9}. 
Our derivation of the Marchenko system for \eqref{1.1} is complementary to the techniques used in \cite{AEU2023a,AEU2023b,E2018,U2025} but not a trivial consequence of those derivations.

We remark that it was 
Kaup and Newell \cite{KN1978} who first derived a system of
linear integral equations for the linear system \eqref{1.7} in the special case 
where $\tilde r(x)=\pm \tilde q(x)^\ast,$ i.e. for the integrable nonlinear system given by
\begin{equation*}
i \tilde q_t+\tilde q_{xx}\mp i (\tilde q^\ast\tilde q^2)_x=0,
\end{equation*}
where we use an asterisk to denote complex conjugation.
Their system of linear integral equations
is the analog of the Marchenko system for \eqref{1.7} in the aforementioned special case.
However, the recovery of the potential $q$ in \cite{KN1978} from the solution of their Marchenko system
is not explicit. This is because the solution \cite{KN1978} to their Marchenko system does not directly yield
the potential $\tilde q$ but instead it yields a product of $\tilde q$ and a term containing the analogs
the function $E(x)$ in \eqref{2.3} and
the constant $\mu$ in \eqref{3.23} in our paper. This is in contrast to the
Marchenko system used in \cite{AE2019,AEU2023a,AEU2023b,E2018}, where
the solution of that Marchenko system explicitly yields
not only $(\tilde q,\tilde r)$ but also explicitly
yields the analogs of the quantities $E(x)$ and $\mu.$

Tsuchida \cite{T2010} formulated a linear system of integral equations to solve the DNLS I system \eqref{1.6},
where $\tilde q$ and $\tilde r$ are 
recovered from the solution of that system of integral equations.
The integral kernel in Tsuchida's system of integral equations lacks the symmetry that exists
in the Marchenko systems established in \cite{AE2019,AEU2023a,AEU2023b,E2018,U2025}, and it is unclear how
that kernel is related to the scattering data for \eqref{1.7}. In \cite{AE2019,E2018} an alternate
linear system
of integral equations is derived in the spirit of Tsuchida's system of integral
equations. The kernel of the alternate Marchenko system in \cite{AE2019,E2018} has the appropriate symmetry and that 
the kernel and the nonhomogeneous terms in that alternate Marchenko system are related to a certain integral
of the Fourier transform
of the scattering data associated with \eqref{1.7}.

In the analysis of the direct scattering problem, we introduce the four particular solutions to
\eqref{1.1} known as the Jost solutions.
We then introduce the scattering coefficients for \eqref{1.1} by
using the spacial asymptotics of those Jost solutions.
In order to establish the pertinent properties of the Jost solutions and
scattering coefficients for \eqref{1.1}, we relate those quantities to
the corresponding quantities for the linear system \eqref{1.7}.
This is because the appearance of the two potentials $\tilde q$ and $\tilde r$ in \eqref{1.7}
is simpler than the appearance of the two potentials $q$ and $r$ in \eqref{1.1}.
This allows us to establish various properties such as the analyticity, continuity, spacial asymptotics, and
spectral asymptotics related to \eqref{1.1} with the help of the corresponding properties
related to \eqref{1.7}.

We also relate the Jost solutions and scattering coefficients for
\eqref{1.1} to the corresponding quantities associated with the respective two linear systems
given by
\begin{equation}\label{1.11}
\displaystyle\frac{d}{dx}\begin{bmatrix}
\xi\\
\noalign{\medskip}
\eta
\end{bmatrix}=
\begin{bmatrix}
-i\lambda & u(x)
\\
\noalign{\medskip}
v(x)&i\lambda
\end{bmatrix}
\begin{bmatrix}
\xi\\
\noalign{\medskip}
\eta
\end{bmatrix},\qquad x\in\mathbb R,
\end{equation}
\begin{equation}\label{1.12}
\displaystyle\frac{d}{dx}\begin{bmatrix}
\gamma\\
\noalign{\medskip}
\epsilon
\end{bmatrix}=
\begin{bmatrix}
-i\lambda & p(x)
\\
\noalign{\medskip}
s(x)&i\lambda
\end{bmatrix}
\begin{bmatrix}
\gamma\\
\noalign{\medskip}
\epsilon
\end{bmatrix},\qquad x\in\mathbb R,
\end{equation}
where the spectral parameter 
 $\lambda$ is related to the spectral parameter $\zeta$ in \eqref{1.1} as 
\begin{equation}\label{1.13}
\lambda=\zeta^2, \quad \zeta=\sqrt{\lambda},
\end{equation}
with the square root denoting the principal part of the complex-valued square root function. 
Each of the two linear systems in \eqref{1.11} and \eqref{1.12} is an AKNS system \cite{AKNS1974} related to
the integrable system of NLS (nonlinear Schr\"odinger) equations, where the NLS system with the potential pair
$(u,v)$ is given by
\begin{equation*}
\begin{cases}
i u_t+u_{xx}-2u^2 v=0, \\
i v_t-v_{xx}+2uv^2 =0.
\end{cases}
\end{equation*}
The appearance of $(u,v)$ and $(p,s)$ in \eqref{1.11} and \eqref{1.12}, respectively,
is simpler than the appearance of $(q,r)$ in \eqref{1.1}. 
We observe from \eqref{1.11} and \eqref{1.12} that there is no dependence on the
spectral parameter in the appearance of $(u,v)$ and $(p,s).$
Consequently, the properties related to the analyticity, continuity, spacial asymptotics, and
spectral asymptotics for \eqref{1.11} and \eqref{1.12} are easier to determine than the corresponding properties 
related to \eqref{1.1} and \eqref{1.7}. This helps us establish those properties
related to \eqref{1.1} by exploiting the corresponding properties related to \eqref{1.11} and \eqref{1.12}.

Our paper is organized as follows. In Section~\ref{section2} we relate the linear system \eqref{1.1} to each of 
the linear systems \eqref{1.7}, \eqref{1.11}, and \eqref{1.12}, respectively,
by choosing 
$(\tilde q,\tilde r),$ $(u,v),$ and $(p,s)$ in a particular way
in terms of $(q,r)$ in \eqref{1.1}. Our particular choices are motivated
by the simplicity of the corresponding connections between
the Jost solutions to \eqref{1.1} and the
Jost solutions to each of \eqref{1.7}, \eqref{1.11}, and \eqref{1.12}.
In Section~\ref{section3} we introduce the Jost solutions and scattering
coefficients for \eqref{1.1}.
The scattering coefficients are introduced by using the spacial asymptotics
of the Jost solutions. Since the coefficient matrix
in \eqref{1.1} has zero trace, the scattering coefficients
for \eqref{1.1} can also be obtained by using certain Wronskians of the
Jost solutions to \eqref{1.1}.
The Jost solutions and scattering coefficients for each of
\eqref{1.7}, \eqref{1.11}, and \eqref{1.12}
are defined in the same manner the corresponding quantities are defined for \eqref{1.1}.
In Section~\ref{section3} we also relate the Jost solutions and scattering coefficients
for \eqref{1.1} to the corresponding quantities for
each of
\eqref{1.7}, \eqref{1.11}, and \eqref{1.12}.
This allows us to establish the pertinent properties of
the Jost solutions and scattering coefficients
for \eqref{1.1} with the help of the known properties of the corresponding
quantities for
each of
\eqref{1.7}, \eqref{1.11}, and \eqref{1.12}.
In Section~\ref{section4} we introduce the bound-state information
for \eqref{1.1} in terms of the matrix triplet pair
$(A,B,C)$ and $(\bar A,\bar B,\bar C).$
The use of matrix triplets not only allows us to deal with any number of bound states with any multiplicities
but also provides the bound-state information as the most suitable input to the Marchenko system.
In Section~\ref{section5} we introduce our system of Marchenko integral equations associated with
\eqref{1.1}. This is done by first formulating the inverse problem for
\eqref{1.1} as a Riemann-Hilbert problem and then modifying that problem appropriately so
that the Fourier transformation is applied and the corresponding Marchenko system
is derived.
In Section~\ref{section5} we also describe how $(q,r)$ is recovered from
the solution of the Marchenko system.
In Section~\ref{section6} we consider the special case when the reflection coefficients for
\eqref{1.1} are zero. We present the explicit solution of our Marchenko system for \eqref{1.1}
in the reflectionless case. We provide two illustrative examples to recover 
$(q,r)$ when the input scattering data set is specified as
a matrix triplet pair.
Finally, in Section~\ref{section7} we conclude our paper by summarizing the significance of our present work 
and by mentioning our plan for the relevant future work.

\section{Relationships among the four linear systems}
\label{section2}

In this section, we relate
the linear systems \eqref{1.1}
to each of the linear systems
\eqref{1.7}, \eqref{1.11}, and \eqref{1.12}.
In each case, we first present a general relationship and then 
choose a particular relationship so that the corresponding potential pairs 
are connected to each other in the simplest manner.

In the next theorem we relate \eqref{1.1} to \eqref{1.7}
by showing how a solution of \eqref{1.1} is related
to a solution of \eqref{1.7} and
how $(q,r)$ is related to $(\tilde q,\tilde r).$

\begin{theorem}
\label{theorem2.1}
The solutions to the linear system \eqref{1.1} and the solutions to
the linear system \eqref{1.7} are related to each other as	
\begin{equation}
\label{2.1}
\begin{bmatrix}\alpha\\
\noalign{\medskip}
\beta
\end{bmatrix}=\begin{bmatrix}
a_1\,E(x)^{-1}& 0\\
\noalign{\medskip}
0&a_2\,E(x)
\end{bmatrix}\begin{bmatrix}
\tilde\alpha\\
\noalign{\medskip}
\tilde\beta
\end{bmatrix},
\qquad x\in\mathbb R,
\end{equation}
where $a_1$ and $a_2$ are arbitrary complex constants, 
$(q,r)$ in \eqref{1.1} is related to
$(\tilde q,\tilde r)$ in \eqref{1.7} as
\begin{equation}
\label{2.2}
q(x)=\displaystyle\frac{a_1}{a_2}\,\tilde q(x)\,E(x)^{-2}, \quad r(x)=\displaystyle\frac{a_2}{a_1}\,\tilde r(x)\,E(x)^2,
\end{equation}
and the complex-valued scalar quantity $E(x)$ is 
given by
\begin{equation}\label{2.3}
E(x):=\exp\left(\displaystyle\frac{i}{2}\displaystyle\int_{-\infty}^x dy\, q(y)\,r(y)\right).
\end{equation}
Thus, $(q,r)$ belongs to the Schwartz class  $\mathcal S(\mathbb R)$ if and only if $(\tilde q,\tilde r)$ belongs to  $\mathcal S(\mathbb R).$
\end{theorem}

\begin{proof}
We relate the solutions to \eqref{1.1} and \eqref{1.7} to each other via a $2\times 2$ matrix denoted by $F$ as
\begin{equation}
\label{2.4}
\begin{bmatrix}
\alpha\\
\noalign{\medskip}
\beta
\end{bmatrix}=F\begin{bmatrix}
\tilde\alpha\\
\noalign{\medskip}
\tilde\beta
\end{bmatrix},
\end{equation}	
where we assume that each entry of $F$ is a function of $x$ and is independent of
the spectral parameter $\zeta.$ Even though the assumption of $\zeta$-independence
puts a constraint on the choices for $F,$ our proof shows that there still exist
matrices $F$ satisying our assumption.
By taking the $x$-derivative of both sides of \eqref{2.4}
and using \eqref{1.1} and \eqref{2.4} on the left-hand side of the resulting equality, we obtain
\begin{equation}
\label{2.5}
\begin{bmatrix}
-i\zeta^2- \displaystyle\frac{i}{2} \,qr & \zeta q\\
\noalign{\medskip}
\zeta r & i\zeta^2+\displaystyle\frac{i}{2} \,qr
\end{bmatrix}F\begin{bmatrix}
\tilde\alpha\\
\noalign{\medskip}
\tilde\beta
\end{bmatrix}=
F'\begin{bmatrix}
\tilde\alpha\\
\noalign{\medskip}
\tilde\beta
\end{bmatrix}+F\begin{bmatrix}
\tilde\alpha\\
\noalign{\medskip}
\tilde\beta
\end{bmatrix}',
\end{equation}	
where we use a prime to denote the $x$-derivative. We let $F_1,$ $F_2,$ $F_3,$ and $F_4$ denote
the entries of the matrix $F,$ i.e. we use
\begin{equation}
\label{2.6}
F=\begin{bmatrix}
F_1& F_2\\
\noalign{\medskip}
F_3&F_4
\end{bmatrix}.
\end{equation}
We remark that \eqref{2.5} must hold for any $\begin{bmatrix} \tilde\alpha\\
\tilde\beta\end{bmatrix}.$ Hence, with the help of \eqref{1.7} and \eqref{2.6} we write \eqref{2.5} in the equivalent form as
\begin{equation*}
\begin{bmatrix}
F'_1&F'_2\\
\noalign{\medskip}
F'_3&F'_4
\end{bmatrix}=\begin{bmatrix}
 -\displaystyle\frac{i}{2} \,q\,r\,F_1+\zeta(\,q\,F_3+\tilde{r}\,F_2)& -\displaystyle\frac{i}{2} \,q\,r\,F_2+ \zeta (q\,F_4-\tilde{q}\,F_1)\\
\noalign{\medskip}
\zeta (r\,F_1-\tilde{r}\,F_4)+F_3(2i\zeta^2+\displaystyle\frac{i}{2} \,q\,r) & \zeta (r\,F_2-\tilde{q}\,F_3)+\displaystyle\frac{i}{2} \,q\,r\,F_4
\end{bmatrix},
\end{equation*}	
where each of the four equalities associated with the corresponding entries is a polynomial equation
in $\zeta$ with degree $1.$ Since $F$ is assumed to be independent of
$\zeta,$ in each equality the coefficients of $\zeta$ on both sides should match and the constant 
terms should also match. This yields the six scalar equalities given by
\begin{equation}
\label{2.8}
\begin{cases}
F'_1= -\displaystyle\frac{i}{2} \,qr\,F_1,\quad F_2=0, \quad F_3=0,
\\ \noalign{\medskip}
F'_4= \displaystyle\frac{i}{2} \,qr\,F_4,\quad q\,F_4=\tilde q\,F_1, \quad r\,F_1=\tilde r\,F_4.
\end{cases}
\end{equation}
We would like to solve \eqref{2.8} to obtain the six quantities $F_1$, $F_2$, $F_3$, $F_4$, $\tilde q,$ $\tilde r$ in
terms of $q$ and $r.$ To solve the equality involving $F_4'$ in \eqref{2.8}, we introduce the quantity
$E(x)$ as the unique solution of the initial-value problem given by
\begin{equation}
\label{2.9}E'(x)=\frac{i}{2} \,q(x)\,r(x)\,E(x),\quad E(-\infty)=1.
\end{equation}
 In fact, the unique solution of \eqref{2.9} is given by the right-hand side of \eqref{2.3}. Thus, we determine
$F_4$ as
$F_4=a_2 E(x),$ where $a_2$ is an arbitrary complex constant.
Then, the equality in \eqref{2.8} involving $F_1'$ has the general solution given by
$F_1=a_1 E(x)^{-1},$ where $a_1$ is an arbitrary complex constant.
Thus, when $(q,r)$ is known,
the general solution of \eqref{2.8} is given by
\begin{equation}
\label{2.10}
\begin{cases}
F_1= a_1\,E(x)^{-1},\quad F_2=0, \quad F_3=0,
\\ \noalign{\medskip}
F_4=  a_2\,E(x),\quad \tilde q(x)=\displaystyle\frac{a_2}{a_1}\,q(x)\,E(x)^2, \quad \tilde r(x)=\displaystyle\frac{a_1}{a_2} \, r(x)\,E(x)^{-2}.
\end{cases}
\end{equation}
As seen from \eqref{2.10}, we observe that the matrix $F$ becomes independent of $\zeta$ if we choose
the constants $a_1$ and $a_2$ as independent of $\zeta.$
Using the entries of $F$ given in \eqref{2.10} as input to \eqref{2.4}, we obtain \eqref{2.1}.
The last two equalities in the second line of \eqref{2.10} yield
\eqref{2.2}, and the quantity $E(x)$ in \eqref{2.3} corresponds to the
unique solution of \eqref{2.9}. Hence, the proof is complete. \end{proof}

From \eqref{2.2}, we observe that the potential pair $(q,r)$ corresponds 
to a one-parameter family of potential pairs 
$(\tilde q,\tilde r)$ parametrized by 
the complex parameter 
$a_1/a_2.$ Motivated by simplicity, by 
letting $a_1=a_2$ we choose the particular potential pair 
$(\tilde q,\tilde r)$ in \eqref{1.7} so that we
have the connection between $(\tilde q,\tilde r)$ in \eqref{1.7}
and $(q,r)$ in \eqref{1.1} given by
\begin{equation}\label{2.11}
\tilde q(x)=E(x)^2\,q(x),\quad \tilde r(x)=E(x)^{-2}\,r(x).
\end{equation}
Without loss of generality, from now on we assume that 
$(\tilde q,\tilde r)$ in \eqref{1.7} is related to
$(q,r)$ in \eqref{1.1} as in \eqref{2.11}.
When \eqref{2.11} holds, from Theorem~\ref{theorem2.1} we obtain
the following result relating the solutions to \eqref{1.1} and 
\eqref{1.7} to each other.

\begin{corollary}
\label{corollary2.2}
Suppose that the potentials $q$ and $r$ in \eqref{1.1} belong to the Schwartz class
$\mathcal S(\mathbb R).$ Let the potentials 
$\tilde q$ and $\tilde r$ appearing in \eqref{1.7}
be related to $q$ and $r$ as in \eqref{2.11}, where $E(x)$ is the complex-valued quantity in \eqref{2.3}.
Then, the potentials $\tilde q$ and $\tilde r$ 
also belong to $\mathcal S(\mathbb R).$ Furthermore, any solution $\begin{bmatrix}
\alpha\\
\beta
\end{bmatrix}$ to \eqref{1.1} 
and any solution $\begin{bmatrix}
\tilde\alpha\\
\tilde\beta
\end{bmatrix}$ to \eqref{1.7} are related to each other as
\begin{equation}
\label{2.12}
\begin{bmatrix}
\alpha\\
\noalign{\medskip}
\beta
\end{bmatrix}=a\begin{bmatrix}
E(x)^{-1}& 0\\
\noalign{\medskip}
0&E(x)
\end{bmatrix}\begin{bmatrix}
\tilde\alpha\\
\noalign{\medskip}
\tilde\beta
\end{bmatrix},
\end{equation}
where $a$ is an arbitrary complex constant.
\end{corollary}

Next we establish the connection between the solutions to the linear systems \eqref{1.1} and \eqref{1.11},
respectively,
in the spirit of Theorem~\ref{theorem2.1}.

\begin{theorem}
\label{theorem2.3}
The solutions $\begin{bmatrix}
\alpha\\
\beta
\end{bmatrix}$ to 
\eqref{1.1} and the solutions $\begin{bmatrix}
\xi\\
\eta
\end{bmatrix}$ to \eqref{1.11} are related to each other as
\begin{equation}
\label{2.13}
\begin{bmatrix}
\alpha\\
\noalign{\medskip}
\beta
\end{bmatrix}=\begin{bmatrix}
b_1& 0\\
\noalign{\medskip}
-b_1\displaystyle\frac{r(x)}{2i\zeta}&\displaystyle\frac{b_2}{\zeta}
\end{bmatrix}\begin{bmatrix}
\xi\\
\noalign{\medskip}
\eta
\end{bmatrix},
\end{equation}
where $b_1$  and $b_2$ are arbitrary complex constants,
$(q,r)$ in \eqref{1.1} is related to $(u,v)$ in \eqref{1.11}
as
 \begin{equation}\label{2.14}
u(x)=\displaystyle\frac{b_2}{b_1}\,q(x), \quad v(x)=\displaystyle\frac{b_1}{b_2}\left[-\displaystyle\frac{i\,r'(x)}{2}-\displaystyle\frac{q(x)\,r(x)^2}{4}\right].
\end{equation}
It follows directly from \eqref{2.14} that
$(u,v)$ belongs to the Schwartz class $\mathcal S(\mathbb R)$
if and only if
$(q,r)$ belongs to $\mathcal S(\mathbb R).$
\end{theorem}

\begin{proof}
The idea of the proof can be found in \cite{AE2019,E2018,T2010}.
We premultiply both sides of \eqref{1.1} by the $2\times 2$ constant diagonal matrix
$\text{\rm{diag}}\{1,\zeta\},$ 
where we recall that $\zeta$ is the spectral parameter in \eqref{1.1}.
The resulting matrix equality can be written as		
\begin{equation}\label{2.15}
\displaystyle\begin{bmatrix}
\alpha\\
\noalign{\medskip}
\zeta\beta
\end{bmatrix}'=
\begin{bmatrix}
-i\zeta^2- \displaystyle\frac{i}{2} \,q(x)\,r(x) & q(x)\\
\noalign{\medskip}
\zeta^2 r(x) & i\zeta^2+\displaystyle\frac{i}{2} \,q(x)\,r(x)
\end{bmatrix}
\begin{bmatrix}
\alpha\\
\noalign{\medskip}
\zeta\beta
\end{bmatrix},\qquad x\in\mathbb R.
\end{equation}
We introduce the $2\times 2$ matrix $G$
to connect the modified wavefunction
$\begin{bmatrix}
\alpha\\
\zeta\beta
\end{bmatrix}$ in \eqref{2.15} and the wavefunction 
$\begin{bmatrix}
\xi\\
\noalign{\medskip}
\eta
\end{bmatrix}$
by letting
\begin{equation}
\label{2.16}
\begin{bmatrix}
\alpha\\
\noalign{\medskip}
\zeta\beta
\end{bmatrix}=G\begin{bmatrix}
\xi\\
\noalign{\medskip}
\eta
\end{bmatrix}.
\end{equation}	
We use $G_1,$ $G_2,$ $G_3,$ $G_4$ to denote the entries of $G,$ 
i.e. we let
\begin{equation*}
G=\begin{bmatrix}
G_1& G_2\\
\noalign{\medskip}
G_3&G_4
\end{bmatrix}.
\end{equation*}
Contrary to the $\zeta$-independence assumption for the matrix $F$ appearing in \eqref{2.6}, we
cannot impose the restriction of $\zeta$-independence on the choice
for the matrix $G.$
By taking the $x$-derivative of both sides of \eqref{2.16}
and using \eqref{2.15} in the resulting matrix equality,
we get
\begin{equation}
\label{2.18}
\begin{bmatrix}
-i\zeta^2- \displaystyle\frac{i}{2} \,q(x)\,r(x) & q(x)\\
\noalign{\medskip}
\zeta^2 r(x) & i\zeta^2+\displaystyle\frac{i}{2} \,q(x)\,r(x)
\end{bmatrix}G\begin{bmatrix}
\xi\\
\noalign{\medskip}
\eta
\end{bmatrix}=
G'\begin{bmatrix}
\xi\\
\noalign{\medskip}
\eta
\end{bmatrix}+G\begin{bmatrix}
\xi\\
\noalign{\medskip}
\eta
\end{bmatrix}'.
\end{equation}	
Next, we use \eqref{1.11} in the second term on the right-hand side of
\eqref{2.18}, and we write the resulting matrix equality in terms of the entries of $G$ as
\begin{equation}
\label{2.19}
\begin{bmatrix}
G'_1&G'_2\\
\noalign{\medskip}
G'_3&G'_4
\end{bmatrix}=\begin{bmatrix}
-\displaystyle\frac{i}{2} \,q\,r\,G_1+q\,G_3-v\,G_2& -2i\zeta^2\,G_2-\displaystyle\frac{i}{2} \,q\,r\,G_2+ q\,G_4-u\,G_1\\
\noalign{\medskip}
2i\zeta^2\,G_3+\zeta^2 r\,G_1+\displaystyle\frac{i}{2} \,q\,r\,G_3-v\,G_4 & \zeta r\,G_2+\displaystyle\frac{i}{2} \,q\,r\,G_4-u\,G_3
\end{bmatrix}.
\end{equation}	
We remark that 
the right-hand side of \eqref{2.19} contains $\zeta$ and $\zeta^2$ in the coefficients.
Viewing \eqref{2.19} as a system of four polynomial equalities
in $\zeta,$ we obtain six scalar equations with 
the six unknown quantities $G_1,$ $G_2,$ $G_3,$ $G_4,$ $u,$ $v$ to be determined in terms of $(q,r).$ We have
\begin{equation}
\label{2.20}
\begin{cases}
G'_1= -\displaystyle\frac{i}{2} \,q\,r\,G_1+q\,G_3,\quad G_2=0, \quad G'_3=\displaystyle\frac{i}{2} \,q\,r\,G_3-v\,G_4,
\\ \noalign{\medskip}
2iG_3+r\,G_1=0, \quad 
G'_4= \displaystyle\frac{i}{2} \,q\,r\,G_4-u\,G_3,\quad u\,G_1=q\,G_4.
\end{cases}
\end{equation}
The general solution of
\eqref{2.20} is given by
\begin{equation}
\label{2.21}
\begin{cases}
G_1= b_1,\quad G_2=0, \quad G_3=-b_1\displaystyle\frac{r(x)}{2i},
\\ \noalign{\medskip}
G_4= b_2,\quad 
u(x)=\displaystyle\frac{b_2}{b_1}\,q(x), \quad v(x)=\displaystyle\frac{b_1}{b_2}\left[-\displaystyle\frac{i\,r'(x)}{2}-\displaystyle\frac{q(x)\,r(x)^2}{4}\right],
\end{cases}
\end{equation}
where $b_1$ and $b_2$ are arbitrary complex constants.
From \eqref{2.21} we get \eqref{2.14} and 
we also obtain the matrix equality
\begin{equation}
\label{2.22}
\begin{bmatrix}
G_1& G_2\\
\noalign{\medskip}
G_3&G_4
\end{bmatrix}=\begin{bmatrix}
b_1& 0\\
\noalign{\medskip}
-b_1\displaystyle\frac{r(x)}{2i}&b_2
\end{bmatrix}.
\end{equation}
Using \eqref{2.22} in \eqref{2.16}
and premultiplying the resulting matrix equality
by the inverse of the diagonal matrix
$\text{\rm{diag}}\{1,\zeta\},$ we obtain \eqref{2.13}.
From \eqref{2.14} it is seen that
$(q,r)$ belongs to
$\mathcal S(\mathbb R)$
 if and only if $(u,v)$
belongs to
$\mathcal S(\mathbb R).$
Hence, the proof of the theorem is complete. 
\end{proof}

From \eqref{2.14}, we observe that the potential pair $(q,r)$ corresponds 
to a one-parameter family of potential pairs 
$(u,v)$ parametrized by 
the complex parameter 
$b_1/b_2.$ Motivated by simplicity, by 
letting $b_1=b_2$ in \eqref{2.21} we choose the particular potential pair 
$(u,v)$ in \eqref{1.11} so that we
have the connection between $(u,v)$ in \eqref{1.11}
and $(q,r)$ in \eqref{1.1} given by
\begin{equation}\label{2.23}
u(x)=q(x),\quad  v(x)=-\displaystyle\frac{i\,r'(x)}{2}-\displaystyle\frac{q(x)\,r(x)^2}{4}.
\end{equation}
Without loss of generality, from now on we assume that
$(u,v)$ in \eqref{1.11} is related to
$(q,r)$ in \eqref{1.1} as in \eqref{2.23}.

When $(u,v)$ 
in \eqref{1.11} and $(q,r)$
in \eqref{1.1} are related to each other as in \eqref{2.23}, from Theorem~\ref{theorem2.3} we obtain
the following result relating the respective solutions to \eqref{1.1} and \eqref{1.11}.

\begin{corollary}
\label{corollary2.4}
Assume that the potentials $q$ and $r$ in \eqref{1.1} belong to the Schwartz class 
$\mathcal S(\mathbb R).$ 
Let $(u,v)$ in \eqref{1.11} be related
to $(q,r)$ 
as in \eqref{2.23}.
Then, $(u,v)$ also belongs to
$\mathcal S(\mathbb R).$ 
Furthermore, any solution $\begin{bmatrix}
\alpha\\
\beta
\end{bmatrix}$ to \eqref{1.1} and any solution $\begin{bmatrix}
\xi\\
\eta
\end{bmatrix}$ to \eqref{1.11} are related 
to each other as
\begin{equation}
\label{2.24}
\begin{bmatrix}
\alpha\\
\noalign{\medskip}
\beta
\end{bmatrix}=b\begin{bmatrix}
1& 0\\
\noalign{\medskip}
-\displaystyle\frac{r(x)}{2i\zeta}&\displaystyle\frac{1}{\zeta}
\end{bmatrix}\begin{bmatrix}
\xi\\
\noalign{\medskip}
\eta
\end{bmatrix},
\end{equation}
where $b$  is an arbitrary complex constant.
\end{corollary}

Finally, we relate the solutions to \eqref{1.1} to the solutions to \eqref{1.12}
by proceeding in a manner similar to the way we have connected
\eqref{1.1} and \eqref{1.11} in Theorem~\ref{theorem2.3} and Corollary~\ref{corollary2.4}.
Next, we present the analog of Theorem~\ref{theorem2.3}.

\begin{theorem}
\label{theorem2.5}
The solutions to the linear system \eqref{1.1} and the
solutions to the linear system \eqref{1.12} are related to each other as
\begin{equation}
\label{2.25}
\begin{bmatrix}
\alpha\\
\noalign{\medskip}
\beta
\end{bmatrix}=\begin{bmatrix}
\displaystyle\frac{c_1}{\zeta}& c_1\displaystyle\frac{q(x)}{2i\zeta}\\
\noalign{\medskip}
0&c_2
\end{bmatrix}\begin{bmatrix}
\gamma\\
\noalign{\medskip}
\epsilon
\end{bmatrix},
\end{equation}
where $c_1$ and $c_2$ are arbitrary complex constants,
and $(q,r)$ in \eqref{1.1} is related to
$(p,s)$ in \eqref{1.12} as
\begin{equation}
\label{2.26}
p(x)=\displaystyle\frac{c_2}{c_1}\left[\displaystyle\frac{i\,q'(x)}{2}
-\displaystyle\frac{q(x)^2\,r(x)}{4}\right],\quad
s(x)=\displaystyle\frac{c_1}{c_2}\,r(x).
\end{equation}
Consequently, 
 $(p,s)$ belongs to the Schwartz class
$\mathcal S(\mathbb R)$ 
if and only if
$(q,r)$ belongs to $\mathcal S(\mathbb R).$ 
\end{theorem}

\begin{proof}
The basic idea behind the proof is similar to the proof of Theorem~\ref{theorem2.3}.
We premultiply both sides of \eqref{1.1} by the $2\times 2$ constant diagonal matrix
$\text{\rm{diag}}\{\zeta,1\},$ 
where we recall that $\zeta$ is the spectral parameter in \eqref{1.1}.
The resulting matrix equality is given by	
\begin{equation}\label{2.27}
\displaystyle\begin{bmatrix}
\zeta\alpha\\
\noalign{\medskip}
\beta
\end{bmatrix}'=
\begin{bmatrix}
-i\zeta^2- \displaystyle\frac{i}{2} \,q(x)\,r(x) & \zeta^2q(x)\\
\noalign{\medskip}
 r(x) & i\zeta^2+\displaystyle\frac{i}{2} \,q(x)\,r(x)
\end{bmatrix}
\begin{bmatrix}
\zeta\alpha\\
\noalign{\medskip}
\beta
\end{bmatrix},\qquad x\in\mathbb R.
\end{equation}
We then introduce the $2\times 2$ matrix $H$ 
to connect the modified wavefunction
$\begin{bmatrix}
\zeta\alpha\\
\beta
\end{bmatrix}$ in \eqref{2.27} and the wavefunction 
$\begin{bmatrix}
\gamma\\
\noalign{\medskip}
\epsilon
\end{bmatrix}$ in \eqref{1.12}
by letting
\begin{equation}
\label{2.28}
\begin{bmatrix}
\zeta\alpha\\
\noalign{\medskip}
\beta
\end{bmatrix}=H\begin{bmatrix}
\gamma\\
\noalign{\medskip}
\epsilon
\end{bmatrix}.
\end{equation}	
We use $H_1,$ $H_2,$ $H_3,$ $H_4$ to denote the entries of $H,$ 
i.e. we let
\begin{equation*}
H=\begin{bmatrix}
H_1& H_2\\
\noalign{\medskip}
H_3&H_4
\end{bmatrix}.
\end{equation*}
We remark that, contrary to the $\zeta$-independence assumption for the matrix $F$ appearing in \eqref{2.6}, we
cannot impose the restriction of $\zeta$-independence on the choice
for the matrix $H.$
By taking the $x$-derivative of both sides of \eqref{2.28}
and using \eqref{2.27} on the left-hand side of the resulting matrix equality,
we get
\begin{equation}
\label{2.30}
\begin{bmatrix}
-i\zeta^2- \displaystyle\frac{i}{2} \,q(x)\,r(x) & \zeta^2  q(x)\\
\noalign{\medskip}
r(x) & i\zeta^2+\displaystyle\frac{i}{2} \,q(x)\,r(x)
\end{bmatrix}H\begin{bmatrix}
\gamma\\
\noalign{\medskip}
\epsilon
\end{bmatrix}=
H'\begin{bmatrix}
\gamma\\
\noalign{\medskip}
\epsilon
\end{bmatrix}+H\begin{bmatrix}
\gamma\\
\noalign{\medskip}
\epsilon
\end{bmatrix}'.
\end{equation}	
Next, we use \eqref{1.12} in the second term on the right-hand side of
\eqref{2.30}, and we write the resulting matrix equality in terms of the entries of $H$ as
\begin{equation}
\label{2.31}
\begin{bmatrix}
H'_1&H'_2\\
\noalign{\medskip}
H'_3&H'_4
\end{bmatrix}=\begin{bmatrix}
-\displaystyle\frac{i}{2} \,q\,r\,H_1+\zeta\,q\,H_3-s\,H_2& -2i\zeta^2\,H_2-\displaystyle\frac{i}{2} \,q\,r\,H_2+ \zeta^2 q\,H_4-p\,H_1\\
\noalign{\medskip}
2i\zeta^2\,H_3+r\,H_1+\displaystyle\frac{i}{2} \,q\,r\,H_3-s\,H_4 & r\,H_2+\displaystyle\frac{i}{2} \,q\,r\,H_4-p\,H_3
\end{bmatrix}.
\end{equation}	
Viewing \eqref{2.31} as a system of four polynomial equalities
in $\zeta,$ we obtain six scalar equations with 
the six unknown quantities $H_1,$ $H_2,$ $H_3,$ $H_4,$ $p,$ $s$ to be determined in terms of $(q,r).$ We have
\begin{equation}
\label{2.32}
\begin{cases}
H'_1= -\displaystyle\frac{i}{2} \,q\,r\,H_1-s\,H_2,\quad H_3=0, \quad H'_2=-\displaystyle\frac{i}{2} \,q\,r\,H_3-p\,H_1,
\\ \noalign{\medskip}
-2iH_2+q\,H_4=0,\quad 
H'_4= \displaystyle\frac{i}{2} \,q\,r\,H_4+r\,H_4,\quad r\,H_1=s\,H_4.
\end{cases}
\end{equation}
The general solution of
\eqref{2.32} is given by
\begin{equation}
\label{2.33}
\begin{cases}
H_1= c_1,\quad H_2=c_1\displaystyle\frac{q(x)}{2i}, \quad H_3=0,
\\ \noalign{\medskip}
H_4= c_2,\quad 
s(x)=\displaystyle\frac{c_1}{c_2}\,r(x), \quad 
p(x)=\displaystyle\frac{c_2}{c_1}\left[\displaystyle\frac{i\,q'(x)}{2}-\displaystyle\frac{q(x)^2\,r(x)}{4}\right],
\end{cases}
\end{equation}
where $c_1$ and $c_2$ are arbitrary complex constants. From \eqref{2.33}, we get
\eqref{2.26} and the matrix equality
\begin{equation}
\label{2.34}
\begin{bmatrix}
H_1& H_2\\
\noalign{\medskip}
H_3&H_4
\end{bmatrix}=\begin{bmatrix}
c_1& c_1\displaystyle\frac{q(x)}{2i}\\
\noalign{\medskip}
0&c_2
\end{bmatrix}.
\end{equation}
Using \eqref{2.34} in \eqref{2.28}
and premultiplying the resulting matrix equality
by the inverse of the diagonal matrix
$\text{\rm{diag}}\{\zeta,1\},$ we obtain \eqref{2.25}.
From \eqref{2.26} it is seen that
$(p,s)$ belongs to
$\mathcal S(\mathbb R)$
 if and only if $(q,r)$
belongs to
$\mathcal S(\mathbb R).$
Hence, the proof of the theorem is complete. 
\end{proof}

From \eqref{2.26}, we observe that the potential pair $(q,r)$ corresponds 
to a one-parameter family of potential pairs 
$(p,s)$ parametrized by 
the complex parameter 
$c_1/c_2.$ Motivated by simplicity, by 
letting $c_1=c_2$ in \eqref{2.33} we choose the particular potential pair 
$(p,s)$ in \eqref{1.12} so that we
have the connection between $(p,s)$ in \eqref{1.12}
and $(q,r)$ in \eqref{1.1} given by
\begin{equation}\label{2.35}
p(x)=\displaystyle\frac{i\,q'(x)}{2}-\displaystyle\frac{q(x)^2\,r(x)}{4},\quad
s(x)=r(x). 
\end{equation}
Without loss of generality, from now on we assume that 
$(p,s)$ in \eqref{1.12} is related to
$(q,r)$ in \eqref{1.1} as in \eqref{2.35}.

When $(p,s)$ 
in \eqref{1.12} and $(q,r)$
in \eqref{1.1} are related to each other as in \eqref{2.35}, from Theorem~\ref{theorem2.5} we obtain
the following corollary relating the respective solutions to \eqref{1.1} and \eqref{1.12}.

\begin{corollary}
\label{corollary2.6}
Assume that the potentials $q$ and $r$ in \eqref{1.1} belong to the Schwartz class 
$\mathcal S(\mathbb R).$ 
Let $(p,s)$ in \eqref{1.12} be related
to $(q,r)$ 
as in \eqref{2.35}.
Then, $(p,s)$ also belongs to
$\mathcal S(\mathbb R).$ 
Furthermore, any solution $\begin{bmatrix}
\alpha\\
\beta
\end{bmatrix}$
to \eqref{1.1} and any solution $\begin{bmatrix}
\gamma\\
\noalign{\medskip}
\epsilon
\end{bmatrix}$
to \eqref{1.12} are related 
to each other as
\begin{equation}
\label{2.36}
\begin{bmatrix}
\alpha\\
\noalign{\medskip}
\beta
\end{bmatrix}=c\begin{bmatrix}
\displaystyle\frac{1}{\zeta}& \displaystyle\frac{q(x)}{2i\zeta}\\
\noalign{\medskip}
0&1
\end{bmatrix}\begin{bmatrix}
\gamma\\
\noalign{\medskip}
\epsilon
\end{bmatrix},
\end{equation}
where $c$  is an arbitrary complex constant.
\end{corollary}

Let us address the issue of relating the linear system \eqref{1.1} to two 
different linear systems given in \eqref{1.11} and \eqref{1.12}, respectively, rather than 
relating it to only one of these two AKNS systems.
As seen from the first equality in \eqref{2.23}, the potentials $u(x)$ and $q(x)$ are related to each other in a simple manner,
and the second equality of \eqref{2.35} shows that the potentials $r(x)$ and $s(x)$ are related to each other also in a simple manner.
On the other hand, we see from \eqref{2.23} that, if we want to express the potential
$r(x)$ in terms of $u(x)$ and $v(x),$ not only we have to use both of $u(x)$ and $v(x)$ but we also have to solve a Riccati
equation. Similarly, we see from \eqref{2.35} that we cannot express $q(x)$ in terms of $p(x)$ and $s(x)$ in a simple manner,
and instead we must use both $p(x)$ and $s(x)$ and we further must solve a Riccati equation.
Hence, it is more advantageous to relate \eqref{1.1} to both \eqref{1.11} and \eqref{1.12} rather than to only one of those two AKNS systems.

\section{The Jost solutions and scattering coefficients}
\label{section3}

In this section we describe the Jost solutions
and the scattering coefficients for \eqref{1.1}
and present their pertinent properties. Those properties are needed later on to establish the Marchenko method to
solve the inverse problem \eqref{1.1}.
To obtain the pertinent properties of the Jost solutions and scattering coefficients, we use
the results from Section~\ref{section2} and relate
the Jost solutions to \eqref{1.1} to
the Jost solutions to each of the linear systems \eqref{1.7}, \eqref{1.11}, and \eqref{1.12}.
We recall that the potential pairs
$(\tilde q,\tilde r),$ $(u,v),$ $(p,s)$ appearing in
\eqref{1.7}, \eqref{1.11}, \eqref{1.12}, respectively,
are related to the potential pair $(q,r)$ in \eqref{1.1} 
as in \eqref{2.11}, \eqref{2.23}, \eqref{2.35}, respectively

We already know from Section~\ref{section2} that $(\tilde q,\tilde r)$ in \eqref{1.7}, $(u,v)$ in \eqref{1.11},
and $(p,s)$ in \eqref{1.12} each belong to
the Schwartz class $\mathcal S(\mathbb R)$
because we assume that the potentials $q$ and $r$ in
\eqref{1.1} belong $\mathcal S(\mathbb R).$
As $x\to\pm\infty,$ each of the linear systems \eqref{1.1}, \eqref{1.7}, \eqref{1.11}, and \eqref{1.12}
reduce to the same unperturbed linear system given by
\begin{equation}\label{3.1}
\displaystyle\frac{d}{dx}\begin{bmatrix}
\overset\circ\alpha\\
\noalign{\medskip}
\overset\circ\beta
\end{bmatrix}=
\begin{bmatrix}
-i\zeta^2&0
\\
\noalign{\medskip}
0&i\zeta^2
\end{bmatrix}
\begin{bmatrix}
\overset\circ\alpha\\
\noalign{\medskip}
\overset\circ\beta
\end{bmatrix},\qquad x\in\mathbb R.
\end{equation}
The general solution of \eqref{3.1} is a linear combination of the column-vector
solutions
$\begin{bmatrix}
e^{-i\zeta^2x}\\
0
\end{bmatrix}$ and
$\begin{bmatrix}0\\
e^{i\zeta^2x}
\end{bmatrix}.$
Consequently, the spacial asymptotics of the scattering solutions
to each of the linear systems \eqref{1.1}, \eqref{1.7}, \eqref{1.11}, and \eqref{1.12}
can be treated in the same manner. In particular, the Jost solutions and
the scattering coefficients to those four linear systems can be defined in the same
manner by using the two spacial asymptotics $\begin{bmatrix}
e^{-i\zeta^2x}\\
0
\end{bmatrix}$ and
$\begin{bmatrix}0\\
e^{i\zeta^2x}
\end{bmatrix}$ for solutions
to those four linear systems.

We first introduce the four Jost solutions to \eqref{1.1},
denoted by
 $\psi(\zeta,x),$ $\bar\psi(\zeta,x),$ $\phi(\zeta,x),$ $\bar\phi(\zeta,x),$ respectively.
We use the subscripts $1$ and $2$ to identify the respective first
and second components of
the Jost solutions, i.e. we let
\begin{equation} \label{3.2}
\psi(\zeta,x)=\begin{bmatrix}
\psi_1(\zeta,x)\\
\noalign{\medskip}\psi_2(\zeta,x)
\end{bmatrix},\quad \bar\psi(\zeta,x)=\begin{bmatrix}
\bar\psi_1(\zeta,x)\\ \noalign{\medskip}\bar\psi_2(\zeta,x)
\end{bmatrix},
\end{equation}
\begin{equation} \label{3.3}
\phi(\zeta,x)=\begin{bmatrix}
\phi_1(\zeta,x)\\
\noalign{\medskip}\phi_2(\zeta,x)
\end{bmatrix},\quad \bar\phi(\zeta,x)=\begin{bmatrix}
\bar\phi_1(\zeta,x)\\ \noalign{\medskip}\bar\phi_2(\zeta,x)
\end{bmatrix}.
\end{equation}
The Jost solutions to \eqref{1.1} are the solution satisfying the respective spacial asymptotics
\begin{equation}\label{3.4}
\begin{bmatrix}
\psi_1(\zeta,x)\\
\noalign{\medskip}\psi_2(\zeta,x)
\end{bmatrix}=\begin{bmatrix}
o(1)\\
\noalign{\medskip}
 e^{i\zeta^2x}\left[1+o(1)\right]
\end{bmatrix} ,\qquad  x\to+\infty,
\end{equation}
\begin{equation}\label{3.5}
\begin{bmatrix}
\bar\psi_1(\zeta,x)\\ \noalign{\medskip}\bar\psi_2(\zeta,x)
\end{bmatrix}=\begin{bmatrix}
e^{-i\zeta^2x}\left[1+o(1)\right]\\
\noalign{\medskip}
o(1)
\end{bmatrix} ,\qquad  x\to+\infty,
\end{equation}
\begin{equation}\label{3.6}
\begin{bmatrix}
\phi_1(\zeta,x)\\
\noalign{\medskip}\phi_2(\zeta,x)
\end{bmatrix}=\begin{bmatrix}
e^{-i\zeta^2x}\left[1+o(1)\right]\\
\noalign{\medskip}
o(1)
\end{bmatrix} ,\qquad   x\to-\infty,
\end{equation}
\begin{equation}\label{3.7}
\begin{bmatrix}
\bar\phi_1(\zeta,x)\\ \noalign{\medskip}\bar\phi_2(\zeta,x)
\end{bmatrix}=\begin{bmatrix}
o(1)\\
\noalign{\medskip}
e^{i\zeta^2x}\left[1+o(1)\right]
\end{bmatrix} ,\qquad  x\to-\infty.
\end{equation}

Next, we introduce the scattering coefficients associated with the linear system
\eqref{1.1} by using the spacial asymptotics of
the Jost solutions to \eqref{1.1} as
\begin{equation}\label{3.8}
\begin{bmatrix}
\psi_1(\zeta,x)\\
\noalign{\medskip}\psi_2(\zeta,x)
\end{bmatrix}=\begin{bmatrix}
\displaystyle\frac{L(\zeta)}{T(\zeta)}\,e^{-i\zeta^2 x}\left[1+o(1)\right]\\
\noalign{\medskip}
\displaystyle\frac{1}{T(\zeta)}\,e^{i\zeta^2 x}\left[1+o(1)\right]
\end{bmatrix}, \qquad   x\to-\infty,
\end{equation}
\begin{equation}\label{3.9}
\begin{bmatrix}
\bar\psi_1(\zeta,x)\\ \noalign{\medskip}\bar\psi_2(\zeta,x)
\end{bmatrix}=\begin{bmatrix}
\displaystyle\frac{1}{\bar T(\zeta)}\,e^{-i\zeta^2 x}\left[1+o(1)\right]\\
\noalign{\medskip}
\displaystyle\frac{\bar L(\zeta)}{\bar T(\zeta)}\,e^{i\zeta^2 x}\left[1+o(1)\right]
\end{bmatrix}, \qquad  x\to-\infty,
\end{equation}
\begin{equation}\label{3.10}
\begin{bmatrix}
\phi_1(\zeta,x)\\
\noalign{\medskip}\phi_2(\zeta,x)
\end{bmatrix}=\begin{bmatrix}
\displaystyle\frac{1}{T(\zeta)}\,e^{-i\zeta^2x}\left[1+o(1)\right]\\
\noalign{\medskip}
\displaystyle\frac{R(\zeta)}{T(\zeta)}\,e^{i\zeta^2 x}\left[1+o(1)\right]
\end{bmatrix}, \qquad   x\to+\infty,
\end{equation}
\begin{equation}\label{3.11}
\begin{bmatrix}
\bar\phi_1(\zeta,x)\\ \noalign{\medskip}\bar\phi_2(\zeta,x)
\end{bmatrix}=\begin{bmatrix}
\displaystyle\frac{\bar R(\zeta)}   {\bar T(\zeta)}\,e^{-i\zeta^2 x}\left[1+o(1)\right]\\
\noalign{\medskip}
\displaystyle\frac{1}{\bar T(\zeta)}\,e^{i\zeta^2 x}\left[1+o(1)\right]
\end{bmatrix}, \qquad   x\to+\infty.
\end{equation}
We refer to $T(\zeta)$ and $\bar T(\zeta)$ as the transmission coefficients,
$L(\zeta)$ and $\bar L(\zeta)$ as the left reflection coefficients, and
$R(\zeta)$ and $\bar R(\zeta)$ as the right reflection coefficients. Since the trace of the coefficient 
matrix in \eqref{1.1} is zero, the transmission coefficients from the left and right coincide, and hence
we only use the two symbols
$T(\zeta)$ and $\bar T(\zeta)$ to denote the
transmission
coefficients.

Alternatively, the scattering coefficients for \eqref{1.1} can be introduced via certain Wronskians 
of the Jost solutions. The Wronskian of any two 
solutions to \eqref{1.1} does not depend on $x,$ due to the fact that the coefficient 
matrix in \eqref{1.1} has the zero trace. For any two column-vector solutions  $\begin{bmatrix}
\alpha_1\\
\beta_1
\end{bmatrix} $ and $\begin{bmatrix}
\alpha_2\\
\beta_2
\end{bmatrix} $ to \eqref{1.1}, the Wronskian is given by
\begin{equation}\label{3.12}
\begin{bmatrix}
\begin{bmatrix}
\alpha_1\\
\noalign{\medskip}
\beta_2
\end{bmatrix};\begin{bmatrix}
\alpha_2\\
\noalign{\medskip}
\beta_2
\end{bmatrix}
\end{bmatrix}:=\begin{vmatrix}
\alpha_1&\alpha_2\\
\noalign{\medskip}
\beta_1&\beta_2
\end{vmatrix},
\end{equation}
where the absolute-value bars in \eqref{3.12} are used to denote the determinant of 
a $2\times2$ matrix. By evaluating certain Wronskians of the Jost solutions to \eqref{1.1} 
as $x\to\pm\infty$
and by using \eqref{3.4}--\eqref{3.11}, we get
the equalities
\begin{equation}\label{3.13}
\begin{bmatrix}
\psi(\zeta,x);\phi(\zeta,x)
\end{bmatrix}=-\displaystyle\frac{1}{T(\zeta)},
\end{equation}
\begin{equation}\label{3.14}
\begin{bmatrix}
\bar\psi(\zeta,x);\bar{\phi}(\zeta,x)
\end{bmatrix}=\displaystyle\frac{1}{\bar T(\zeta)},
\end{equation}
\begin{equation}\label{3.15}
\begin{bmatrix}
\psi(\zeta,x);\bar{\phi}(\zeta,x)
\end{bmatrix}=-\displaystyle\frac{\bar R(\zeta)}{\bar T(\zeta)}=\displaystyle\frac{L(\zeta)}{T(\zeta)},
\end{equation}
\begin{equation}\label{3.16}
\begin{bmatrix}
\bar\psi(\zeta,x);\phi(\zeta,x)
\end{bmatrix}=\displaystyle\frac{R(\zeta)}{T(\zeta)}=-\displaystyle\frac{\bar L(\zeta)}{\bar T(\zeta)}.
\end{equation}
Hence, using \eqref{3.13}--\eqref{3.16} we express the scattering coefficients via certain Wronskians as
\begin{equation}\label{3.17}
T(\zeta)=\displaystyle\frac{1}{\begin{bmatrix}
\phi(\zeta,x);\psi(\zeta,x)
\end{bmatrix}}, \quad \bar T(\zeta)=\displaystyle\frac{1}{\begin{bmatrix}
\bar\psi(\zeta,x);\bar{\phi}(\zeta,x)
\end{bmatrix}},
\end{equation}
\begin{equation}\label{3.18}
R(\zeta)=\displaystyle\frac{\begin{bmatrix}
\phi(\zeta,x);\bar\psi(\zeta,x)
\end{bmatrix}}{\begin{bmatrix}
\psi(\zeta,x);\phi(\zeta,x)
\end{bmatrix}}, \quad \bar R(\zeta)=\displaystyle\frac{\begin{bmatrix}
\bar{\phi}(\zeta,x);\psi(\zeta,x)
\end{bmatrix}}{\begin{bmatrix}
\bar\psi(\zeta,x);\bar{\phi}(\zeta,x)
\end{bmatrix}},
\end{equation}
\begin{equation}\label{3.19}
L(\zeta)=\displaystyle\frac{\begin{bmatrix}
\psi(\zeta,x);\bar{\phi}(\zeta,x)
\end{bmatrix}}{\begin{bmatrix}
\phi(\zeta,x);\psi(\zeta,x)
\end{bmatrix}}, \quad \bar L(\zeta)=\displaystyle\frac{\begin{bmatrix}
\phi(\zeta,x);\bar\psi(\zeta,x)
\end{bmatrix}}{\begin{bmatrix}
\bar\psi(\zeta,x);\bar{\phi}(\zeta,x)
\end{bmatrix}}.
\end{equation}
Let us remark that from \eqref{3.13}--\eqref{3.16} we see that the left and right reflection coefficients for \eqref{1.1}  satisfy
\begin{equation}\label{3.20}
L(\zeta)=-\displaystyle \frac{\bar R(\zeta)\,T(\zeta)}{\bar T(\zeta)}, \quad \bar L(\zeta)=-\displaystyle\frac{R(\zeta)\,\bar T(\zeta)}{T(\zeta)},
\end{equation}
\begin{equation}\label{3.21}	T(\zeta)\,\bar T(\zeta)
=1-L(\zeta)\,\bar L(\zeta)=1-R(\zeta)\,\bar R(\zeta).
\end{equation}

In order to denote the Jost solutions 
to \eqref{1.7}, we use the respective notations
 $\psi^{(\tilde q,\tilde r)}(\zeta,x),$ $\bar\psi^{(\tilde q,\tilde r)}(\zeta,x),$ 
$\phi^{(\tilde q,\tilde r)}(\zeta,x),$ $\bar\phi^{(\tilde q,\tilde r)}(\zeta,x)$ by indicating the corresponding potentials in the superscripts.  
They are the solutions to \eqref{1.7} satisfying the respective spacial asymptotics given in
\eqref{3.4}--\eqref{3.7}.
We again use the subscripts $1$ and $2$ to identify their
first and second components.
For the scattering coefficients for \eqref{1.7},
we use
$T^{(\tilde q,\tilde r)}(\zeta)$ and $\bar T^{(\tilde q,\tilde r)}(\zeta)$ to denote the transmission coefficients,
$L^{(\tilde q,\tilde r)}(\zeta)$ and $\bar L^{(\tilde q,\tilde r)}(\zeta)$ for the left reflection coefficients, and
$R^{(\tilde q,\tilde r)}(\zeta)$ and $\bar R^{(\tilde q,\tilde r)}(\zeta)$ 
for the right reflection coefficients.
Those scattering coefficients are obtained from the asymptotics of the corresponding Jost solutions as in
\eqref{3.8}--\eqref{3.11}. Since the coefficient matrix in \eqref{1.7} has zero trace,
those scattering coefficients can alternatively be introduced by using the Wronskians of the corresponding Jost solutions to \eqref{1.7}.
Hence, they also satisfy
the analogs of \eqref{3.13}--\eqref{3.21}.

To denote the Jost solutions 
to \eqref{1.11}, we use
 $\psi^{(u,v)}(\lambda,x),$ $\bar\psi^{(u,v)}(\lambda,x),$ $\phi^{(u,v)}(\lambda,x),$ $\bar\phi^{(u,v)}(\lambda,x),$ respectively. 
We recall that $\lambda$ is related to $\zeta$ as in \eqref{1.13}.
Those Jost solutions are the solutions to \eqref{1.11} satisfying the respective spacial asymptotics given in
\eqref{3.4}--\eqref{3.7}.
We again use the subscripts $1$ and $2$ to identify their
first and second components.
As the scattering coefficients for \eqref{1.11},
we use
 $T^{(u,v)}(\lambda)$ and $\bar T^{(u,v)}(\lambda)$ to denote the transmission coefficients,
$L^{(u,v)}(\lambda)$ and $\bar L^{(u,v)}(\lambda)$ for the left reflection coefficients, and
$R^{(u,v)}(\lambda)$ and $\bar R^{(u,v)}(\lambda)$
for the right reflection coefficients.
Those scattering coefficients are obtained from the asymptotics of the corresponding Jost solutions as in
\eqref{3.8}--\eqref{3.11}. Since the coefficient matrix in \eqref{1.11} has the zero trace, the scattering coefficients for \eqref{1.11}
can alternatively be introduced by using \eqref{3.17}--\eqref{3.19} via the Wronskians of the corresponding Jost solutions to \eqref{1.11}.
Those scattering coefficients also satisfy the analogs of \eqref{3.13}--\eqref{3.21}.

To denote the Jost solutions 
to \eqref{1.12}, we use
 $\psi^{(p,s)}(\lambda,x),$ $\bar\psi^{(p,s)}(\lambda,x),$ $\phi^{(p,s)}(\lambda,x),$ $\bar\phi^{(p,s)}(\lambda,x),$ respectively. 
 Those Jost solutions are the solutions to \eqref{1.11} satisfying the respective spacial asymptotics given in
 \eqref{3.4}--\eqref{3.7}.
 We again use the subscripts $1$ and $2$ to identify their
 first and second components.
As the scattering coefficients for \eqref{1.12},
we use
 $T^{(p,s)}(\lambda)$ and $\bar T^{(p,s)}(\lambda)$ to denote the transmission coefficients,
$L^{(p,s)}(\lambda)$ and $\bar L^{(p,s)}(\lambda)$ for the left reflection coefficients, and
$R^{(p,s)}(\lambda)$ and $\bar R^{(p,s)}(\lambda)$
for the right reflection coefficients.
Those scattering coefficients are obtained from the asymptotics of the corresponding Jost solutions as in
\eqref{3.8}--\eqref{3.11}.
Since the coefficient matrix in \eqref{1.12} has the zero trace, the scattering coefficients for \eqref{1.12}
can alternatively be introduced by using the Wronskians of the corresponding Jost solutions to \eqref{1.12}.
Those scattering coefficients also satisfy the analogs of \eqref{3.13}--\eqref{3.21}.

In the following theorem, we present the connection between the Jost solutions to \eqref{1.1} and
the Jost solutions to \eqref{1.7} when $(q,r)$ in \eqref{1.1} and $(\tilde q,\tilde r)$ in \eqref{1.7} are related to each other as in \eqref{2.11}.

\begin{theorem}
\label{theorem3.1}
Suppose  that $(q,r)$  in \eqref{1.1} belongs to the Schwartz class $\mathcal S(\mathbb R)$ and that
it is related to $(\tilde q,\tilde r)$ as in \eqref{2.11}, where
$E(x)$ is the quantity defined in \eqref{2.3}.
Then, we have the following:

\begin{enumerate}

\item[\text{\rm(a)}]
The Jost solution $\psi(\zeta,x)$ 
to \eqref{1.1} is related to the Jost solution $\psi^{(\tilde q,\tilde r)}(\zeta,x)$ to \eqref{1.7} as
\begin{equation}\label{3.22}
\begin{bmatrix}
\psi_1(\zeta,x)\\
\noalign{\medskip}\psi_2(\zeta,x)
\end{bmatrix}=e^{-i \mu/2}\begin{bmatrix}
\displaystyle\frac{1}{E(x)}& 0\\
\noalign{\medskip}
0&E(x)
\end{bmatrix}\begin{bmatrix}
\psi^{(\tilde q,\tilde r)}_1(\zeta,x)\\
\noalign{\medskip}\psi^{(\tilde q,\tilde r)}_2(\zeta,x)
\end{bmatrix},
\end{equation}
where the quantities $\psi^{(\tilde q,\tilde r)}_1(\zeta,x)$ and $\psi^{(\tilde q,\tilde r)}_2(\zeta,x)$ denote the respective components of the Jost solution
$\psi^{(\tilde q,\tilde r)}(\zeta,x)$ and the scalar constant $\mu$ is defined as
\begin{equation}
\label{3.23}
\mu:=\int_{-\infty}^\infty dy\,q(y)\,r(y).
\end{equation}

\item[\text{\rm(b)}]
The Jost solution $\bar\psi(\zeta,x)$ 
to \eqref{1.1} is related to the Jost solution $\bar\psi^{(\tilde q,\tilde r)}(\zeta,x)$ to \eqref{1.7} as
\begin{equation}\label{3.24}
\begin{bmatrix}
\bar\psi_1(\zeta,x)\\ \noalign{\medskip}\bar\psi_2(\zeta,x)
\end{bmatrix}=e^{i \mu/2}\begin{bmatrix}
\displaystyle\frac{1}{E(x)}& 0\\
\noalign{\medskip}
0&E(x)
\end{bmatrix}\begin{bmatrix}
\bar\psi^{(\tilde q,\tilde r)}_1(\zeta,x)\\ \noalign{\medskip}\bar\psi^{(\tilde q,\tilde r)}_2(\zeta,x)
\end{bmatrix},
\end{equation}
where $\bar\psi^{(\tilde q,\tilde r)}_1(\zeta,x)$ and $\bar\psi^{(\tilde q,\tilde r)}_2(\zeta,x)$ denote the respective components of 
$\bar\psi^{(\tilde q,\tilde r)}(\zeta,x).$

\item[\text{\rm(c)}]
The Jost solution $\phi(\zeta,x)$ 
to \eqref{1.1} is related to the Jost solution $\phi^{(\tilde q,\tilde r)}(\zeta,x)$ to \eqref{1.7} as
\begin{equation}\label{3.25}
\begin{bmatrix}
\phi_1(\zeta,x)\\
\noalign{\medskip}\phi_2(\zeta,x)
\end{bmatrix}=\begin{bmatrix}
\displaystyle\frac{1}{E(x)}& 0\\
\noalign{\medskip}
0&E(x)
\end{bmatrix}\begin{bmatrix}
\phi^{(\tilde q,\tilde r)}_1(\zeta,x)\\
\noalign{\medskip}\phi^{(\tilde q,\tilde r)}_2(\zeta,x)
\end{bmatrix},
\end{equation}
where $\phi^{(\tilde q,\tilde r)}_1(\zeta,x)$ and $\phi^{(\tilde q,\tilde r)}_2(\zeta,x)$ denote the respective components of 
$\phi^{(\tilde q,\tilde r)}(\zeta,x).$

\item[\text{\rm(d)}]
The Jost solution $\bar\phi(\zeta,x)$ 
to \eqref{1.1} is related to the Jost solution $\bar\phi^{(\tilde q,\tilde r)}(\zeta,x)$ to \eqref{1.7} as
\begin{equation}\label{3.26}
\begin{bmatrix}
\bar\phi_1(\zeta,x)\\ \noalign{\medskip}\bar\phi_2(\zeta,x)
\end{bmatrix}=\begin{bmatrix}
\displaystyle\frac{1}{E(x)}& 0\\
\noalign{\medskip}
0&E(x)
\end{bmatrix}\begin{bmatrix}
\bar\phi^{(\tilde q,\tilde r)}_1(\zeta,x)\\ \noalign{\medskip}\bar\phi^{(\tilde q,\tilde r)}_2(\zeta,x)
\end{bmatrix},
\end{equation}
where $\bar\phi^{(\tilde q,\tilde r)}_1(\zeta,x)$ and $\bar\phi^{(\tilde q,\tilde r)}_2(\zeta,x)$ denote the respective components of
$\bar\phi^{(\tilde q,\tilde r)}(\zeta,x).$

\end{enumerate}
\end{theorem}

\begin{proof} 
From \eqref{2.3} we get
\begin{equation}
\label{3.27}
\displaystyle\lim_{x\to-\infty}E(x)=1, \quad \displaystyle\lim_{x\to+\infty}E(x)=e^{i\mu/2}.
\end{equation}
To establish \eqref{3.22}, we proceed as follows.
From \eqref{2.12} we have
\begin{equation}\label{3.28}
\begin{bmatrix}
\psi_1(\zeta,x)\\
\noalign{\medskip}\psi_2(\zeta,x)
\end{bmatrix}
=a\begin{bmatrix}
\displaystyle\frac{1}{E(x)}&0\\
\noalign{\medskip}
0&E(x)
\end{bmatrix}
\begin{bmatrix}
\psi^{(\tilde q,\tilde r)}_1(\zeta,x)\\
\noalign{\medskip}
\psi^{(\tilde q,\tilde r)}_2(\zeta,x)
\end{bmatrix},
\end{equation}
where we recall that $a$ is a constant.
We let $x\to+\infty$ in \eqref{3.28}, and we use the second equality of \eqref{3.27} as well as the spacial asymptotics in \eqref{3.4}
for the Jost solutions $\psi(\zeta,x)$ to \eqref{1.1} and
$\psi^{(\tilde q,\tilde r)}(\zeta,x)$ to \eqref{1.7}.
This yields
\begin{equation}\label{3.29}
\begin{bmatrix}
o(1)\\
\noalign{\medskip}
e^{i\zeta^2x}\left[1+o(1)\right]
\end{bmatrix}
=a\begin{bmatrix}
e^{-i\mu/2}&0\\
\noalign{\medskip}
0&e^{i\mu/2}
\end{bmatrix}
\begin{bmatrix}
o(1)\\
\noalign{\medskip}
e^{i\zeta^2x}\left[1+o(1)\right]
\end{bmatrix}, \quad x\to +\infty.
\end{equation}
From \eqref{3.29} we see that the constant $a$ in \eqref{3.28}
is equal to $e^{-i\mu/2}.$ Thus, the proof of (a) is complete.
 The relationships presented in (b), (c), and (d) are obtained in a similar manner with the 
help of \eqref{3.5}--\eqref{3.7}, the analogs of \eqref{3.5}--\eqref{3.7} 
for the linear system \eqref{1.7}, and the spacial asymptotics in \eqref{3.27}. We use the relationship in \eqref{2.12} for the respective Jost solutions to 
\eqref{1.1} and \eqref{1.7}, and in each case 
we determine the explicit value of the constant $a$ appearing in \eqref{2.12} for the corresponding pair of Jost solutions.
\end{proof}

In the following theorem, we describe the connection between the Jost solutions to \eqref{1.1} and
the Jost solutions to \eqref{1.11} when $(q,r)$ in \eqref{1.1} and $(u,v)$ in \eqref{1.11} are related to each other as in \eqref{2.23}. 

\begin{theorem}
\label{theorem3.2}
Suppose that $(q,r)$ in \eqref{1.1} belongs to the Schwartz class $\mathcal S(\mathbb R)$ and that it is related
to $(u,v)$ in \eqref{1.11} as in \eqref{2.23}.
Then, we have the following:
	
\begin{enumerate}
		
\item[\text{\rm(a)}]
The Jost solution $\psi(\zeta,x)$ 
to \eqref{1.1} is related to the Jost solution $\psi^{(u,v)}(\lambda,x)$ to \eqref{1.11} as
\begin{equation}\label{3.30}
\begin{bmatrix}
\psi_1(\zeta,x)\\
\noalign{\medskip}\psi_2(\zeta,x)
\end{bmatrix}=\begin{bmatrix}
\zeta& 0\\
\noalign{\medskip}
-\displaystyle\frac{r(x)}{2i}&1
\end{bmatrix}\begin{bmatrix}
 \psi_1^{(u,v)}(\lambda,x)\\
\noalign{\medskip}\psi_2^{(u,v)}(\lambda,x)
\end{bmatrix},
\end{equation}
where $\psi_1^{(u,v)}(\lambda,x)$ and $\psi_2^{(u,v)}(\lambda,x)$ denote the respective components of 
$\psi^{(u,v)}(\lambda,x)$
and we recall that $\lambda$ and $\zeta$ are related to each other as in \eqref{1.13}.

\item[\text{\rm(b)}]
The Jost solution $\bar\psi(\zeta,x)$ 
to \eqref{1.1} is related to the Jost solution $\bar\psi^{(u,v)}(\lambda,x)$ to \eqref{1.11} as
\begin{equation}\label{3.31}
\begin{bmatrix}
\bar\psi_1(\zeta,x)\\ \noalign{\medskip}\bar\psi_2(\zeta,x)
\end{bmatrix}=\begin{bmatrix}
1& 0\\
\noalign{\medskip}
-\displaystyle\frac{r(x)}{2i\zeta}&\displaystyle\frac{1}{\zeta}
\end{bmatrix}\begin{bmatrix}
\bar\psi_1^{(u,v)}(\lambda,x)\\ \noalign{\medskip}\bar\psi_2^{(u,v)}(\lambda,x)
\end{bmatrix},
\end{equation}
where $\bar\psi_1^{(u,v)}(\lambda,x)$ and $\bar\psi_2^{(u,v)}(\lambda,x)$ denote the respective components of 
$\bar\psi^{(u,v)}(\lambda,x).$

\item[\text{\rm(c)}]
The Jost solution $\phi(\zeta,x)$ 
to \eqref{1.1} is related to the Jost solution $\phi^{(u,v)}(\lambda,x)$ to \eqref{1.11} as
\begin{equation}\label{3.32}
\begin{bmatrix}
\phi_1(\zeta,x)\\
\noalign{\medskip}\phi_2(\zeta,x)
\end{bmatrix}=\begin{bmatrix}
1& 0\\
\noalign{\medskip}
-\displaystyle\frac{r(x)}{2i\zeta}&\displaystyle\frac{1}{\zeta}
\end{bmatrix}\begin{bmatrix}
\phi_1^{(u,v)}(\lambda,x)\\
\noalign{\medskip}\phi_2^{(u,v)}(\lambda,x)
\end{bmatrix},
\end{equation}
where $\phi_1^{(u,v)}(\lambda,x)$ and $\phi_2^{(u,v)}(\lambda,x)$ denote the respective components of
$\phi^{(u,v)}(\lambda,x).$

\item[\text{\rm(d)}]
The Jost solution $\bar\phi(\zeta,x)$ 
to \eqref{1.1} is related to the Jost solution $\bar\phi^{(u,v)}(\lambda,x)$ to \eqref{1.11} as
\begin{equation}\label{3.33}
\begin{bmatrix}
\bar\phi_1(\zeta,x)\\ \noalign{\medskip}\bar\phi_2(\zeta,x)
\end{bmatrix}=\begin{bmatrix}\zeta& 0\\
\noalign{\medskip}
-\displaystyle\frac{r(x)}{2i}&1
\end{bmatrix}\begin{bmatrix}
\bar\phi_1^{(u,v)}(\lambda,x)\\ \noalign{\medskip}\bar\phi_2^{(u,v)}(\lambda,x)
\end{bmatrix},
\end{equation}
where $\bar\phi_1^{(u,v)}(\lambda,x)$ and $\bar\phi_2^{(u,v)}(\lambda,x)$ denote the respective components of 
$\bar\phi^{(u,v)}(\lambda,x).$ 

\end{enumerate}
\end{theorem}

\begin{proof} 
For the proof of (a), we proceed as follows.
From Corollary~\ref{corollary2.4} we know that
the Jost solutions $\psi(\zeta,x)$ and $\psi^{(u,v)}(\lambda,x)$ are related to each other as in \eqref{2.24}, i.e. we have
\begin{equation}\label{3.34}
\begin{bmatrix}
\psi_1(\zeta,x)\\
\noalign{\medskip}\psi_2(\zeta,x)
\end{bmatrix}
=b \begin{bmatrix}
1& 0\\
\noalign{\medskip}
-\displaystyle\frac{r(x)}{2i\zeta}&\displaystyle\frac{1}{\zeta}
\end{bmatrix}
\begin{bmatrix}
\psi_1^{(u,v)}(\lambda,x)\\
\noalign{\medskip}
\psi_2^{(u,v)}(\lambda,x)
\end{bmatrix},
\end{equation}
where $b$ is a constant.
Using the spacial
asymptotics in \eqref{3.4} as
$x\to+\infty$ for
$\psi(\zeta,x)$ and the analog of \eqref{3.4} for the linear system \eqref{1.11}, 
from \eqref{3.34} we get
\begin{equation*}
\begin{bmatrix}
o(1)\\
\noalign{\medskip}
e^{i\zeta^2x}\left[1+o(1)\right]
\end{bmatrix}
=b \begin{bmatrix}
1& 0\\
\noalign{\medskip}
0&\displaystyle\frac{1}{\zeta}
\end{bmatrix}
\begin{bmatrix}
o(1)\\
\noalign{\medskip}
e^{i\zeta^2x}\left[1+o(1)\right]
\end{bmatrix}, \qquad x\to +\infty,
\end{equation*}
from which we get $b=\zeta.$  This establishes \eqref{3.30} and
completes the proof of (a). The relationships presented in (b), (c), and (d) are obtained in a similar 
manner with the  help of the spacial asymptotics in \eqref{3.5}--\eqref{3.7}, and the analogs of 
\eqref{3.5}--\eqref{3.7} for the linear system \eqref{1.11}. 
We use the relationship in \eqref{2.24} for the respective Jost solutions to \eqref{1.1} and \eqref{1.11}, 
determining in each case the explicit value of the constant $b$ appearing in \eqref{2.24} for the corresponding pair of Jost solutions.
\end{proof}

In the following theorem, we establish the connection between the Jost solutions to \eqref{1.1} and
the Jost solutions to \eqref{1.12} when $(p,s)$ in \eqref{1.12} and $(q,r)$ in \eqref{1.1} are related to each other as in \eqref{2.35}. 

\begin{theorem}
\label{theorem3.3}
Assume that $(q,r)$ in \eqref{1.1} belongs to the Schwartz class $\mathcal S(\mathbb R).$ 
Suppose that $(q,r)$ and $(p,s)$ in \eqref{1.12} are related to 
each other as in \eqref{2.35}. Then, we have the following:
	
\begin{enumerate}
		
\item[\text{\rm(a)}]
The Jost solution $\psi(\zeta,x)$ 
to \eqref{1.1} is related to the Jost solution $\psi^{(p,s)}(\lambda,x)$ to \eqref{1.12} as
\begin{equation}\label{3.36}
\begin{bmatrix}
\psi_1(\zeta,x)\\
\noalign{\medskip}\psi_2(\zeta,x)
\end{bmatrix}=\begin{bmatrix}
\displaystyle\frac{1}{\zeta}& \displaystyle\frac{q(x)}{2i\zeta}\\
\noalign{\medskip}
0&1
\end{bmatrix}\begin{bmatrix}
\psi_1^{(p,s)}(\lambda,x)\\
\noalign{\medskip}\psi_2^{(p,s)}(\lambda,x)
\end{bmatrix},
\end{equation}
where $\psi_1^{(p,s)}(\lambda,x)$ and $\psi_2^{(p,s)}(\lambda,x)$ denote the respective components of 
$\psi^{(p,s)}(\lambda,x)$ and
the parameters $\lambda$ and $\zeta$ are related to each other
as in \eqref{1.13}.

\item[\text{\rm(b)}]
The Jost solution $\bar\psi(\zeta,x)$ 
to \eqref{1.1} is related to the Jost solution $\bar\psi^{(p,s)}(\lambda,x)$ to \eqref{1.12} as
\begin{equation}\label{3.37}
\begin{bmatrix}
\bar\psi_1(\zeta,x)\\ \noalign{\medskip}\bar\psi_2(\zeta,x)
\end{bmatrix}=\begin{bmatrix}
1& \displaystyle\frac{q(x)}{2i}\\
\noalign{\medskip}
0&\zeta
\end{bmatrix}\begin{bmatrix}
\bar\psi_1^{(p,s)}(\lambda,x)\\ \noalign{\medskip}\bar\psi_2^{(p,s)}(\lambda,x)
\end{bmatrix},
\end{equation}
where $\bar\psi_1^{(p,s)}(\zeta,x)$ and $\bar\psi_2^{(p,s)}(\zeta,x)$ denote the respective components of 
$\bar\psi^{(p,s)}(\zeta,x).$

\item[\text{\rm(c)}]
The Jost solution $\phi(\zeta,x)$ 
to \eqref{1.1} is related to the Jost solution $\phi^{(p,s)}(\lambda,x)$ to \eqref{1.12} as
\begin{equation}\label{3.38}
\begin{bmatrix}
\phi_1(\zeta,x)\\
\noalign{\medskip}\phi_2(\zeta,x)
\end{bmatrix}=\begin{bmatrix}
1& \displaystyle\frac{q(x)}{2i}\\
\noalign{\medskip}
0&\zeta
\end{bmatrix}\begin{bmatrix}
\phi_1^{(p,s)}(\lambda,x)\\
\noalign{\medskip}\phi_2^{(p,s)}(\lambda,x)
\end{bmatrix},
\end{equation}
where $\phi_1^{(p,s)}(\lambda,x)$ and $\phi_2^{(p,s)}(\lambda,x)$ denote the respective components of 
$\phi^{(p,s)}(\lambda,x).$

\item[\text{\rm(d)}]
The Jost solution $\bar\phi(\zeta,x)$ 
to \eqref{1.1} is related to the Jost solution $\bar\phi^{(p,s)}(\lambda,x)$ to \eqref{1.12} as
\begin{equation}\label{3.39}
\begin{bmatrix}
\bar\phi_1(\zeta,x)\\ \noalign{\medskip}\bar\phi_2(\zeta,x)
\end{bmatrix}=\begin{bmatrix}
\displaystyle\frac{1}{\zeta}& \displaystyle\frac{q(x)}{2i\zeta}\\
\noalign{\medskip}
0&1
\end{bmatrix}\begin{bmatrix}
\bar\phi_1^{(p,s)}(\lambda,x)\\ \noalign{\medskip}\bar\phi_2^{(p,s)}(\lambda,x)
\end{bmatrix},
\end{equation}
where $\bar\phi_1^{(p,s)}(\lambda,x)$ and $\bar\phi_2^{(p,s)}(\lambda,x)$ denote the respective components of 
$\bar\phi^{(p,s)}(\lambda,x).$

\end{enumerate}
\end{theorem}

\begin{proof} 	
We first present the proof of (a).
When $(q,r)$ and $(p,s)$ are related to each other as
in \eqref{2.35}, from \eqref{2.36} we get
\begin{equation}\label{3.40}
\begin{bmatrix}
\psi_1(\zeta,x)\\
\noalign{\medskip}\psi_2(\zeta,x)
\end{bmatrix}
=c\begin{bmatrix}
\displaystyle\frac{1}{\zeta}& \displaystyle\frac{q(x)}{2i\zeta}\\
\noalign{\medskip}
0&1
\end{bmatrix}
\begin{bmatrix}
\psi_1^{(p,s)}(\lambda,x)\\
\noalign{\medskip}
\psi_2^{(p,s)}(\lambda,x)
\end{bmatrix},
\end{equation}
where $c$ is a constant.
With the help of the asymptotics in \eqref{3.4} and the analog of \eqref{3.4} for the linear system \eqref{1.12},
from \eqref{3.40} we have
\begin{equation*}
\begin{bmatrix}
o(1)\\
\noalign{\medskip}
e^{i\zeta^2x}\left[1+o(1)\right]
\end{bmatrix}
=c\begin{bmatrix}
\displaystyle\frac{1}{\zeta}& 0\\
\noalign{\medskip}
0&1
\end{bmatrix}
\begin{bmatrix}
o(1)\\
\noalign{\medskip}
e^{i\zeta^2x}\left[1+o(1)\right]
\end{bmatrix}, \qquad x\to +\infty, 
\end{equation*}
from which we see that
$c=1.$ Thus, we have established \eqref{3.36} and completed the proof of (a).
The relationships presented in (b), (c), and (d) are obtained in a similar manner with the  
help of the spacial asymptotics in \eqref{3.5}--\eqref{3.7} and the analogs of \eqref{3.5}--\eqref{3.7} for the linear system \eqref{1.12}. 
We use the relationship in \eqref{2.36} for the respective Jost solutions of \eqref{1.1} 
and \eqref{1.12}, and in each case we determine the explicit value of the constant $c$ 
appearing in \eqref{2.36} for the corresponding pair of Jost solutions.
\end{proof}

In the next theorem, we describe some pertinent properties of the Jost solutions to \eqref{1.1}.
Those properties pertain to the existence of the Jost solutions,
their domains of continuity in the spectral parameter $\zeta,$ and their domains of 
analyticity in $\zeta$ and $\lambda.$

\begin{theorem}
\label{theorem3.4}
Suppose that the potentials $q$ and $r$ in \eqref{1.1} belong to the Schwartz class $\mathcal S(\mathbb R).$ 
Let the spectral parameter $\zeta$ be related to the parameter $\lambda$ as in \eqref{1.13}. 
Let $\psi(\zeta,x),$ $\phi(\zeta,x),$
$\bar\psi(\zeta,x),$ and $\bar\phi(\zeta,x)$ denote the Jost solutions to
\eqref{1.1} satisfying the respective spacial asymptotics in \eqref{3.4}--\eqref{3.7},
where the subscripts $1$ and $2$ describe the respective components of the Jost solutions as in
\eqref{3.2} and \eqref{3.3}.
Then, we have the following:
	
\begin{enumerate}
		
\item[\text{\rm(a)}] 
For each fixed $x\in\mathbb R,$ 
the Jost solutions $\psi(\zeta,x),$ and $\phi(\zeta,x)$
exist, are analytic in the first and third quadrants in the complex $\zeta$-plane, and are continuous in the closures of those quadrants.
Similarly, for each fixed $x\in\mathbb R,$ the Jost solutions
$\bar\psi(\zeta,x)$ and $\bar\phi(\zeta,x)$ exist, are analytic in the 
second and fourth quadrants in the complex $\zeta$-plane, and are continuous in the closures of those quadrants.
		
\item[\text{\rm(b)}]  
For each fixed $x\in\mathbb R,$ 
the four components 
$\psi_1(\zeta,x),$ $\bar\psi_2(\zeta,x),$ $\bar\phi_1(\zeta,x),$
$\phi_2(\zeta,x)$ are odd in $\zeta$ and the four components 
$\bar\psi_1(\zeta,x),$ $\psi_2(\zeta,x),$ $\phi_1(\zeta,x),$ $\bar\phi_2(\zeta,x)$ are even in $\zeta.$ 

\item[\text{\rm(c)}]  
For each fixed $x\in\mathbb R,$ the four scalar quantities $\psi_1(\zeta,x)/\zeta,$ $\psi_2(\zeta,x),$ 
$\phi_1(\zeta,x),$ $\phi_2(\zeta,x)/\zeta$ are even in 
$\zeta,$ and hence they are functions of $\lambda.$ 
Those four scalar functions of $\lambda$ are analytic in $\lambda \in\mathbb C^+$ and continuous in 
$\lambda\in\overline{\mathbb C^+}.$ 

\item[\text{\rm(d)}]  
For each fixed $x\in\mathbb R,$ the four scalar quantities 
$\bar\psi_1(\zeta,x),$ $\bar\psi_2(\zeta,x)/\zeta,$ $\bar\phi_1(\zeta,x)/\zeta,$ 
$\bar\phi_2(\zeta,x)$ are even in $\zeta,$ and hence they are functions of $\lambda.$ Those four scalar 
functions of $\lambda$ are analytic in $\lambda \in\mathbb C^-$ and continuous in $\lambda\in\overline{\mathbb C^-}.$
		
\end{enumerate}
\end{theorem}

\begin{proof}
The proof is obtained by proceeding as in Theorem 2.2 of \cite{AEU2023a}, where the corresponding results are obtained for 
the linear system \eqref{1.7}.
\end{proof}

The next theorem presents the large $\zeta$-asymptotics of the components of the Jost solutions to \eqref{1.1}.
With the help of Theorem~\ref{theorem3.4},
in the theorem those asymptotics are expressed in terms of $\lambda,$ which is related to $\zeta$ as in \eqref{1.13}.

\begin{theorem}
\label{theorem3.5}
Suppose that the potentials $q$ and $r$ in \eqref{1.1} belong to the Schwartz 
class $\mathcal S(\mathbb R).$ Let the parameter
$\lambda$ be related to the spectral parameter $\zeta$ as in \eqref{1.13}.
Let $\psi(\zeta,x),$ $\bar\psi(\zeta,x),$ $\phi(\zeta,x),$ $\bar\phi(\zeta,x)$ denote the Jost solutions to
\eqref{1.1} satisfying the respective spacial asymptotics in \eqref{3.4}--\eqref{3.6},
where the subscripts $1$ and $2$ describe the respective components of the Jost solutions as in
\eqref{3.2} and \eqref{3.3}.
Then, for each fixed $x\in\mathbb R,$ the Jost solutions $\psi(\zeta,x)$ and $\phi(\zeta,x)$ satisfy the
respective asymptotics as $\lambda\to\infty$ in $\overline{\mathbb C^+}$
that are given by
\begin{equation}\label{3.42}
\displaystyle\frac{\psi_1(\zeta,x)}{\zeta}=
e^{i\lambda x}\left[\displaystyle\frac{q(x)}{2i\lambda} +O\left(\displaystyle\frac{1}{\lambda^2}\right)\right],
\end{equation}
\begin{equation}
\label{3.43}
\psi_2(\zeta,x)=e^{i\lambda x}\left[1+\displaystyle\frac{q(x)\,r(x)}{4\lambda}
-\displaystyle\frac{1}{2i\lambda}\int_x^\infty dy\,\sigma(y)+O\left(\frac{1}{\lambda^2}\right)\right],
\end{equation}
\begin{equation}
\label{3.44}
\phi_1(\zeta,x)=
\displaystyle e^{-i\lambda x} \left[1-\displaystyle\frac{1}{2i\lambda}\int_{-\infty}^x dy\,\sigma(y)
+O\left(\displaystyle\frac{1}{\lambda^{2}}\right)\right],
\end{equation}
\begin{equation}
\label{3.45}
\displaystyle\frac{\phi_2(\zeta,x)}{\zeta}=e^{-i\lambda x}
\left[-\displaystyle\frac{i\,r(x)}{2\lambda}+O\left(\displaystyle\frac{1}{\lambda^2}\right)\right],
\end{equation}
where the complex-valued scalar quantity $\sigma(x)$ is defined as
\begin{equation}\label{3.46}
\sigma(x):=-\displaystyle\frac{i}{2} \,q(x)\, r'(x)- \displaystyle\frac{1}{4} \,q(x)^2\, r(x)^2.
\end{equation}
Similarly, for each fixed $x\in\mathbb R,$ 
the Jost solutions $\bar\psi(\zeta,x)$ and $\bar\phi(\zeta,x)$ satisfy the
respective asymptotics as $\lambda\to\infty$ in $\overline{\mathbb C^-}$
that are given by
\begin{equation}
\label{3.47}
\bar\psi_1(\zeta,x)=
\displaystyle e^{-i\lambda x} 
\left[1+\displaystyle\frac{1}{2i\lambda}\int_x^\infty dy\,\sigma(y)
+O\left(\displaystyle\frac{1}{\lambda^{2}}\right)\right],
\end{equation}
\begin{equation}\label{3.48}
\displaystyle\frac{\bar\psi_2(\zeta,x)}{\zeta}=
e^{-i\lambda x}\left[-\displaystyle\frac{r(x)}{2i\lambda}+O\left(\displaystyle\frac{1}{\lambda^2}\right)\right],
\end{equation}
\begin{equation}
\label{3.49}
\displaystyle\frac{\bar\phi_1(\zeta,x)}{\zeta}=
e^{i\lambda x}\left[\displaystyle\frac{q(x)}{2i\lambda}+O\left(\displaystyle\frac{1}{\lambda^2}\right)\right],
\end{equation}
\begin{equation}
\label{3.50}
\bar\phi_2(\zeta,x)=
e^{i\lambda x}\left[1+\displaystyle\frac{q(x)\,r(x)}{4\lambda}
+\displaystyle\frac{1}{2i\lambda}\int_{-\infty}^{x}dy\,\sigma(y)
+O\left(\frac{1}{\lambda^2}\right)\right].
\end{equation}
	
\end{theorem}

\begin{proof}
The asymptotics in \eqref{3.42}--\eqref{3.45} and
\eqref{3.47}--\eqref{3.50} can be obtained by using \eqref{3.22}--\eqref{3.26} and the already known large $\zeta$-asymptotics
of the Jost solutions to \eqref{1.7}.
The large $\zeta$-asymptotics
of the Jost solutions to \eqref{1.7} are listed in Theorem 2.4 in \cite{AEU2023a},
 where we express the potentials $\tilde q$ and $\tilde r$ appearing in those
asymptotics in terms of $q$ and $r$ with the help of \eqref{2.11}. We remark that
the expression for $\sigma(x)$ expressed in terms of $(\tilde q,\tilde r)$
is given in (2.36) of \cite{AEU2023a} as
\begin{equation}\label{3.51}
\sigma(x):=-\displaystyle\frac{i}{2} \,\tilde q(x)\,\tilde r'(x)+ \displaystyle\frac{1}{4} \,\tilde q(x)^2\, \tilde r(x)^2.
\end{equation}
We note the second terms on the right-hand sides of \eqref{3.46} and \eqref{3.51} differ by a sign.
Using \eqref{2.11} in \eqref{3.51}, we obtain the expression for $\sigma(x)$ given in
\eqref{3.46}
expressed in terms of $(q,r).$
We mention that the proof can also be obtained by using \eqref{3.30}--\eqref{3.33} or 
\eqref{3.36}--\eqref{3.39} and the known large $\lambda$-asymptotics \cite{AS1981,AE2019,E2018} of 
the Jost solutions to \eqref{1.11} and \eqref{1.12} after we express the potentials $u,$ $v,$ $p,$ $s$ appearing in those
asymptotics in terms of $q$ and $r$ with the help of \eqref{2.23} or \eqref{2.35}.
\end{proof}

The next theorem presents the relationships between the scattering coefficients for \eqref{1.1} and
the scattering coefficients for each of \eqref{1.7}, \eqref{1.11}, \eqref{1.12}, respectively.

\begin{theorem}
\label{theorem3.6}
Suppose that $(q,r)$ in \eqref{1.1} belongs to the Schwartz class $\mathcal S(\mathbb R).$ 
Let $(\tilde q,\tilde r)$ in \eqref{1.7} be related
to $(q,r)$ as in \eqref{2.11},
$(u,v)$ in \eqref{1.11}
be related
to $(q,r)$ as in \eqref{2.23},
and $(p,s)$ in \eqref{1.12} be related
to $(q,r)$ as in \eqref{2.35},
Let $\mu$  the quantity in \eqref{5.1}, and
let the parameters $\lambda$ and $\zeta$ be related to each other as in
\eqref{1.13}.
 Then, we have the following:
\begin{enumerate}

\item[\text{\rm(a)}] The six scattering coefficients $T(\zeta),$ 
$\bar T(\zeta),$ $R(\zeta),$ $L(\zeta),$ $\bar R(\zeta),$ $\bar L(\zeta)$ for 
\eqref{1.1}, the six scattering coefficients $T^{(\tilde q,\tilde r)}(\zeta),$ $\bar T^{(\tilde q,\tilde r)}(\zeta),$ 
$R^{(\tilde q,\tilde r)}(\zeta),$ $L^{(\tilde q,\tilde r)}(\zeta),$ 
$\bar R^{(\tilde q,\tilde r)}(\zeta),$ $\bar L^{(\tilde q,\tilde r)}(\zeta)$ for \eqref{1.7},  
the six scattering coefficients $T^{(u,v)}(\lambda),$ $\bar T^{(u,v)}(\lambda),$ $R^{(u,v)}(\lambda),$ $L^{(u,v)}(\lambda),$ 
$\bar R^{(u,v)}(\lambda),$ $\bar L^{(u,v)}(\lambda)$ for \eqref{1.11}, and 
the six scattering coefficients $T^{(p,s)}(\lambda),$ $\bar T^{(p,s)}(\lambda),$ $R^{(p,s)}(\lambda),$ $L^{(p,s)}(\lambda),$ 
$\bar R^{(p,s)}(\lambda),$ $\bar L^{(p,s)}(\lambda)$ for \eqref{1.12} are related to each other as
\begin{equation}\label{3.52}
 T(\zeta)=e^{i \mu/2} \,T^{(\tilde q,\tilde r)}(\zeta)=T^{(u,v)}(\lambda)=T^{(p,s)}(\lambda),
\qquad \lambda\in\overline{\mathbb C^+},
\end{equation}
\begin{equation}\label{3.53}
\bar {T}(\zeta)=e^{-i \mu/2}\,
\bar T^{(\tilde q,\tilde r)}(\zeta)=\bar {T}^{(u,v)}(\lambda)=\bar {T}^{(p,s)}(\lambda),
\qquad \lambda\in\overline{\mathbb C^-},
\end{equation}
\begin{equation}\label{3.54}
\zeta\, R(\zeta)=\zeta\,e^{i\mu} R^{(\tilde q,\tilde r)}(\zeta)=R^{(u,v)}(\lambda)=\lambda\,R^{(p,s)}(\lambda),
\qquad \lambda\in\mathbb R,
\end{equation}
\begin{equation}\label{3.55}
\zeta\,\bar R(\zeta)=
\zeta\,e^{-i  \mu} \,\bar R^{(\tilde q,\tilde r)}(\zeta)=\lambda\,\bar R^{(u,v)}(\lambda)=
\bar R^{(p,s)}(\lambda),\qquad \lambda\in\mathbb R,
\end{equation}
\begin{equation}\label{3.56}
\zeta\,L(\zeta)=\zeta\,L^{(\tilde q,\tilde r)}(\zeta)=\lambda\,L^{(u,v)}(\lambda)=
L^{(p,s)}(\lambda),\qquad \lambda\in\mathbb R,
\end{equation}
\begin{equation}\label{3.57}
\zeta\,\bar L(\zeta)=
\zeta\,\bar L^{(\tilde q,\tilde r)}(\zeta)=\bar L^{(u,v)}(\lambda)=\lambda\,\bar L^{(p,s)}(\lambda),\qquad \lambda\in\mathbb R.
\end{equation}

\item[\text{\rm(b)}] The transmission coefficient $T(\zeta)$ for \eqref{1.1} is even in $\zeta,$ and hence it is a 
function of $\lambda.$ As a function of $\lambda,$ the quantity $T(\zeta)$ is meromorphic in 
$\lambda\in\mathbb C^+$ and is continuous in $\lambda\in\overline{\mathbb C^+}$ 
except at the poles causing the meromorphic property in $\mathbb C^+.$ 

\item[\text{\rm(c)}] 
The transmission coefficient
$\bar T(\zeta)$ for \eqref{1.1} is even in $\zeta,$ and
hence it is a function of $\lambda.$ As a function of $\lambda,$ the quantity $\bar T(\zeta)$ is 
meromorphic in $\lambda\in\mathbb C^-$ and is continuous in $\lambda\in\overline{\mathbb C^-}$ 
except at the poles causing the meromorphic property in $\mathbb C^-.$ 

\item[\text{\rm(d)}] 
The four quantities
$R(\zeta)/\zeta,$ $\bar R(\zeta)/\zeta,$ $L(\zeta)/\zeta,$ $\bar L(\zeta)/\zeta$ associated with
\eqref{1.1} are even in $\zeta,$ and hence they are all 
functions of $\lambda.$ As functions of $\lambda,$ those four quantities are continuous in 
$\lambda\in\mathbb R.$ 
\end{enumerate}
\end{theorem}

\begin{proof}
For the proof of (a), we proceed as follows. We use the asymptotics \eqref{3.8}--\eqref{3.11} for the 
scattering coefficients for \eqref{1.1} and their
analogs for the 
scattering coefficients for \eqref{1.7}. We then use those asymptotics in 
\eqref{3.22}--\eqref{3.26}, respectively.
By comparing the leading terms in the resulting asymptotic equalities, we obtain the first
equalities in \eqref{3.52}--\eqref{3.57}. The second and third equalities in
\eqref{3.52}--\eqref{3.57} are similarly established with the help of the
analogs of \eqref{3.8}--\eqref{3.11} for \eqref{1.11} and \eqref{1.12}, respectively,
and then using those asymptotics in \eqref{3.30}--\eqref{3.33} and
in \eqref{3.36}--\eqref{3.39}, respectively.
Hence, the proof of (a) is complete.
For the proof of (b), we proceed as follows. It is already known \cite{AKNS1974,AEU2023a}
that, when $(u,v)$ in \eqref{1.7} belongs to
the Schwartz class $\mathcal S(\mathbb R),$ the transmission coefficients $T^{(u,v)}(\lambda)$ is
meromorphic in 
$\lambda\in\mathbb C^+$ and is continuous in $\lambda\in\overline{\mathbb C^+}$ 
except at the poles causing the meromorphic property in $\mathbb C^+.$ 
Then, with the help of \eqref{1.13}, from \eqref{3.52} we conclude that
$T(\zeta)$ satisfies the properties stated in (b). Thus, the proof of (b)
is complete. The proof of (c) is obtained as in the proof of (b), i.e. by using
\eqref{3.53} and the fact  \cite{AKNS1974,AEU2023a} that
$\bar T^{(u,v)}(\lambda)$ is meromorphic in $\lambda\in\mathbb C^-$ and is continuous in 
$\lambda\in\overline{\mathbb C^-}.$ 
Finally, the proof of (d) follows from the first equalities in
\eqref{3.54}--\eqref{3.57} and by using 
the fact that
the four reflection coefficients for \eqref{1.7} are continuous in 
$\lambda\in\mathbb R.$ 
For the proof of the continuity of the reflection coefficients for
\eqref{1.7}, we refer the reader to Theorem~2.5(c) of  \cite{AEU2023a}.
\end{proof}

The next theorem presents the small $\zeta$-asymptotics of the scattering coefficients for \eqref{1.1}.
In the theorem, those asymptotics are expressed in terms of $\lambda,$ which is related to $\zeta$ as in \eqref{1.13}.

\begin{theorem}
\label{theorem3.7}
Assume that the potentials $q$ and $r$ in \eqref{1.1} belong to the Schwartz class $\mathcal S(\mathbb R).$ Let the parameter $\lambda$ be related 
to the spectral parameter $\zeta$ as in \eqref{1.13}. Then, the small $\zeta$-asymptotics of the scattering coefficients 
$T(\zeta),$ $\bar T(\zeta),$  $R(\zeta),$ $\bar R(\zeta),$ $L(\zeta),$ and $\bar L(\zeta)$ appearing in 
\eqref{3.8}--\eqref{3.11} are expressed in $\lambda$ as
\begin{equation}
\label{3.58}
T(\zeta)=e^{i\mu/2}\left[1+O(\lambda)\right],\qquad \lambda\to 0
\text{\rm{ in }} \overline{\mathbb C^+},
\end{equation}
\begin{equation}
\label{3.59}
\bar T(\zeta)=e^{-i\mu/2}\left[1+O(\lambda)\right],\qquad \lambda\to 0
\text{\rm{ in }} \overline{\mathbb C^+},
\end{equation}
\begin{equation}
\label{3.60}
\displaystyle\frac{R(\zeta)}{\zeta}=e^{i\mu}\left[\displaystyle\int_{-\infty}^\infty dy\,r(y)\,E(y)^{-2}+O(\lambda)\right], \qquad \lambda\to 0
\text{\rm{ in }} \mathbb R,
\end{equation}
\begin{equation}
\label{3.61}
\displaystyle\frac{\bar R(\zeta)}{\zeta}=e^{-i\mu}\left[\displaystyle\int_{-\infty}^\infty dy\,q(y)\,E(y)^2+O(\lambda)\right], \qquad \lambda\to 0
\text{\rm{ in }} \mathbb R,
\end{equation}
\begin{equation}
\label{3.62}
\displaystyle\frac{L(\zeta)}{\zeta}=-\displaystyle\int_{-\infty}^\infty dy\,q(y)\,E(y)^2+O(\lambda),\qquad \lambda\to 0
\text{\rm{ in }} \mathbb R,
\end{equation}
\begin{equation}
\label{3.63}
\displaystyle\frac{\bar L(\zeta)}{\zeta}=-\displaystyle\int_{-\infty}^\infty dy\,r(y)\,E(y)^{-2}+O(\lambda), \qquad \lambda\to 0
\text{\rm{ in }} \mathbb R,
\end{equation}
where we recall that $E(x)$ is the quantity defined in \eqref{2.3} and $\mu$ is the
constant in \eqref{3.23}.
\end{theorem}

\begin{proof} 
When 
$(\tilde q,\tilde r)$ in \eqref{1.7} belongs to $\mathbb S(\mathbb R),$ the small $\zeta$-asymptotics of 
the scattering coefficients for \eqref{1.7} are known and listed in Theorem~2.5(d) of \cite{AEU2023a}. We use those asymptotics in 
the first equalities of  \eqref{3.52}--\eqref{3.57}, and we establish \eqref{3.58}--\eqref{3.63}.
\end{proof}

The next theorem presents the large $\zeta$-asymptotics of the scattering coefficients for \eqref{1.1}.
In the theorem, those asymptotics are expressed in terms of $\lambda,$ which is related to $\zeta$ as in \eqref{1.13}. 

\begin{theorem}
\label{theorem3.8}
Assume that the potentials $q$ and $r$ in \eqref{1.1} belong to the Schwartz class $\mathcal S(\mathbb R).$ 
Then, the large $\zeta$-asymptotics of the scattering coefficients 
$T(\zeta),$  $\bar T(\zeta),$  $R(\zeta),$ $\bar R(\zeta),$ $L(\zeta),$ $\bar L(\zeta)$ for
\eqref{1.1} are given by
\begin{equation}
\label{3.64}
T(\zeta)=1+O\left(\displaystyle\frac{1}{\lambda}\right),\qquad \lambda\to \infty
\text{\rm{ in }} \overline{\mathbb C^+},
\end{equation}
\begin{equation}
\label{3.65}
\bar T(\zeta)=1+O\left(\displaystyle\frac{1}{\lambda}\right),\qquad \lambda\to \infty
\text{\rm{ in }} \overline{\mathbb C^+},
\end{equation}
\begin{equation}
\label{3.66}
R(\zeta)=O\left(\displaystyle\frac{1}{\zeta^3}\right), \quad \bar R(\zeta)=O\left(\displaystyle\frac{1}{\zeta^3}\right), \qquad \lambda\to\pm \infty,
\end{equation}
\begin{equation}
\label{3.67}
L(\zeta)=O\left(\displaystyle\frac{1}{\zeta^3}\right),\quad \bar L(\zeta)=O\left(\displaystyle\frac{1}{\zeta^3}\right), \qquad \lambda\to\pm \infty,
\end{equation}
where we recall that $\lambda$ and $\zeta$ are
related to each other as in \eqref{1.13}. 
\end{theorem}

\begin{proof}
When 
$(\tilde q,\tilde r)$ in \eqref{1.7} belongs to $\mathbb S(\mathbb R),$ the large $\zeta$-asymptotics of 
the scattering coefficients for \eqref{1.7} are known and listed in
(2.46)--(2.51) of \cite{AEU2023a}. Using those asymptotics in the first
equalities of \eqref{3.52}--\eqref{3.57}, we obtain \eqref{3.64}--\eqref{3.67} as the large
$\zeta$-asymptotics for the scattering coefficients for \eqref{1.1}.
\end{proof}

\section{The bound states}
\label{section4}

The bound states for \eqref{1.1} correspond to square-integrable column-vector solutions to \eqref{1.1}. 
Such solutions can occur only at certain values of the spectral parameter $\zeta.$
In this section we present the basic information on the bound states for \eqref{1.1} when $(q,r)$ there is supposed to belong to the Schwartz class $\mathcal S(\mathbb R).$
In this paper, we assume that $(\tilde q,\tilde r)$
in \eqref{1.7} is related to $(q,r)$ as in \eqref{2.11}.
Hence, 
$(\tilde q,\tilde r)$ also belongs to $\mathcal S(\mathbb R).$
When the potentials $\tilde q$ and $\tilde r$ are in 
$\mathcal S(\mathbb R),$ the basic information for
the bound states for
\eqref{1.7} is available in Section~3 of
\cite{AEU2023a}. This allows us to obtain the basic information on
the bound states for
\eqref{1.1} by exploiting the relationships established in Section~\ref{section3} between
the Jost solutions and transmission coefficients for \eqref{1.1} and the corresponding quantities for
\eqref{1.7}, respectively.

The bound states for
\eqref{1.7} are related to the meromorphic properties of the transmission coefficients in the complex $\zeta$-plane. From
\eqref{3.52} and \eqref{3.53} we know that the transmission coefficients for \eqref{1.1} and the transmission coefficients for \eqref{1.7}
have similar meromorphic properties in the complex $\zeta$-plane. Thus, by 
using the information on the bound states for \eqref{1.7}, we obtain the facts related to the bound states for \eqref{1.1}. 
In the following we provide a summary of the basic facts related to the bound states for 
\eqref{1.1} when $(q,r)$ belongs to the Schwartz class.

\begin{enumerate}

\item[\text{\rm(a)}] The bound states for \eqref{1.1} cannot occur when $\zeta\in\mathbb R.$
This can be seen as follows. With the help of \eqref{2.9} we see that
\eqref{1.1} at $\zeta=0$ has the two linearly independent solutions 
$\begin{bmatrix} E(x)^{-1}\\ 0\end{bmatrix}$
and $\begin{bmatrix} 0\\
E(x)\end{bmatrix}.$
Since we cannot form a square-integrable solution
by using a linear combination of those two solutions, a bound state at $\zeta=0$
cannot occur. A bound state when $\zeta\in\mathbb R\setminus\{0\}$ cannot occur either.
This is because \eqref{1.1} has then the two linearly independent solutions, namely the Jost solutions
$\psi(\zeta,x)$ and $\bar\psi(\zeta,x),$ and from \eqref{3.4} and \eqref{3.5} it follows that
it is impossible to have a square-integrable solution that is a linear combination of those
two Jost solutions when $\zeta\in\mathbb R\setminus\{0\}.$
A bound state for \eqref{1.1} can only occur at a nonreal
complex $\zeta$-value at which the transmission coefficient $T(\zeta)$ has a pole in the first or third quadrant in the
complex $\zeta$-plane or the transmission coefficient $\bar T(\zeta)$ has a pole in the second or fourth quadrant. 
This fact follows from the first equalities in \eqref{3.52} and \eqref{3.53} and the fact
that a bound state for \eqref{1.7} can only occur at a nonreal
complex $\zeta$-value at which the transmission coefficient $T^{(\tilde q,\tilde r)}(\zeta)$ has a pole in the first or third quadrant in the
complex $\zeta$-plane or the transmission coefficient $\bar T^{(\tilde q,\tilde r)}(\zeta)$ 
has a pole in the second or fourth quadrant. We know from Theorem~\ref{theorem3.6}(b) that the transmission coefficients
$T(\zeta)$ and $\bar T(\zeta)$ are even in $\zeta.$ 
Consequently, the bound-state $\zeta$-values for \eqref{1.1} are located symmetrically with 
respect to the origin of the complex $\zeta$-plane. As a result, it is convenient to describe the
bound-state poles of $T(\zeta)$ and $\bar T(\zeta)$ in terms of the parameter $\lambda,$ 
which is related to $\zeta$ as in \eqref{1.13}. 

\item[\text{\rm(b)}] The number $N$ of distinct poles of
$T(\zeta)$ when $\lambda\in\mathbb C^+$ is finite, and we use $\lambda_j$ 
for $1\le j\le N$ to denote those poles. Similarly, the number $\bar N$ of distinct poles of $\bar T(\zeta)$ 
when $\lambda\in\mathbb C^-$ is finite, and we use $\bar \lambda_j$ for $1\le j\le \bar N$ to denote those
poles. We recall that an overbar in our paper does not denote complex conjugation. 
If $T(\zeta)$ has no poles in $\lambda\in\mathbb C^+,$ then we have
$N=0.$ Similarly, if $\bar T(\zeta)$ has no poles in $\lambda\in\mathbb C^-,$ then 
$\bar N=0.$ The multiplicity of each pole of $T(\zeta)$ is finite, and
 we use $m_j$ to denote the multiplicity of the pole at $\lambda=\lambda_j.$
Similarly, the multiplicity of each pole of $\bar T(\zeta)$ is finite, and we use $\bar m_j$ to 
denote the  multiplicity of the pole at $\lambda=\bar\lambda_j.$

\item[\text{\rm(c)}] The bound-state information for \eqref{1.1} can be
presented in terms of  the two sets $\left\{\lambda_j,m_j\right\}_{j=1}^N$ and
$\left\{\bar\lambda_j,\bar m_j\right\}_{j=1}^{\bar N}.$ For each bound state and multiplicity, we 
introduce  a bound-state normalization 
constant. We use the double-indexed constants $c_{jk}$ for $1\le j\le N$ and
$0\le k\le m_j-1$ and the double-indexed constants $\bar c_{jk}$ for $1\le j\le \bar N$ and
$0\le k\le \bar m_j-1$ to denote the
bound-state normalization constants at
$\lambda=\lambda_j$ and $\lambda=\bar\lambda_j,$ respectively.
The construction of $c_{jk}$ and $\bar c_{jk}$ for \eqref{1.1} is similar to the construction 
for the corresponding bound-state normalization constants for
\eqref{1.7}. We refer the reader to \cite{AEU2023a,AEU2023b} for those
constructions in terms of the corresponding transmission coefficients and bound-state dependency constants.
In summary, the bound-state information for \eqref{1.1} can be specified
by using the two bound-state input data sets given by
\begin{equation}
\label{4.1}
\left\{\lambda_j,m_j,\{c_{jk}\}_{k=0}^{m_j-1}\right\}_{j=1}^N,\quad
\left\{\bar\lambda_j,\bar m_j,\{\bar c_{jk}\}_{k=0}^{\bar m_j-1}\right\}_{j=1}^{\bar N}.
\end{equation}

\item[\text{\rm(d)}] To solve the inverse scattering problem for
\eqref{1.1} via the Marchenko method, it is the most convenient to represent the
 bound-state information with the help of a matrix triplet pair denoted by  
$(A,B,C)$ and  $(\bar A,\bar B,\bar C).$ The use of matrix triplets
in the Marchenko method allows us to handle any number of bound states with any 
multiplicities as if we deal only with a pair of simple bound states.

\end{enumerate}

The matrix triplet pair $(A,B,C)$ and $(\bar A,\bar B,\bar C)$ is related
to \eqref{4.1} as follows. 
For the bound state at $\lambda=\lambda_j$ with the multiplicity $m_j$ for $1\le j\le N,$ 
we chose the matrix triplet $(A,B,C)$ as
\begin{equation}\label{4.2}
A=\begin{bmatrix}
A_1&0&\cdots&0&0\\
0&A_2&\cdots&0&0\\
\vdots&\vdots&\ddots&\vdots&\vdots\\
0&0&\cdots&A_{N-1}&0\\
0&0&\cdots&0&A_N
\end{bmatrix},
\quad
B=\begin{bmatrix}
B_1\\
B_2\\
\vdots\\
B_N
\end{bmatrix}, \quad C=\begin{bmatrix}
C_1&C_2&\cdots&C_N
\end{bmatrix},
\end{equation}
where $A$ is a block diagonal matrix, $B$ is a block column vector, and $C$ is a block row vector.
The matrix subtriplet $(A_j,B_j,C_j)$ is given by 
\begin{equation}\label{4.3}
A_j:=\begin{bmatrix}
\lambda_j&1&0&\cdots&0&0\\
0&\lambda_j&1&\cdots&0&0\\
0&0&\lambda_j&\cdots&0&0\\
\vdots&\vdots&\vdots&\ddots&\vdots&\vdots\\
0&0&0&\cdots&\lambda_j&1\\
0&0&0&\dots&0&\lambda_j
\end{bmatrix},\quad 
B_j:=\begin{bmatrix}
0\\ \vdots \\
0\\
1
\end{bmatrix},
\end{equation}
\begin{equation}\label{4.4}
C_j:=\begin{bmatrix}
c_{j(m_j-1)}&c_{j(m_j-2)}&\cdots&c_{j1}&c_{j0}
\end{bmatrix}.
\end{equation}
As seen from \eqref{4.3}, the $m_j\times m_j$ matrix $A_j$ is in the Jordan canonical form 
with $\lambda_j$ appearing in its diagonal entries and
the $m_j\times 1$ matrix 
$B_j$ has the scalar $0$ in the first  $m_j-1$ entries and $1$ in the  $m_j$th entry. 
As seen from \eqref{4.4}, the $1\times m_j$ matrix  $C_j$ contains the bound-state  normalization constants
$c_{jk}$ for $0\le k\le m_j-1$ in its entries. From
\eqref{4.2} we see that the zeros in the block diagonal matrix
$A$ denote the zero matrices of appropriate matrix sizes.
The matrix size of $A$ is given by $\mathcal N\times \mathcal N,$  where
the nonnegative integer $\mathcal N$ is defined as
\begin{equation}\label{4.5}
\mathcal N:=\displaystyle\sum_{j=1}^{N} m_j.
\end{equation}
As seen from \eqref{4.2}, the matrix size of $B$ is $\mathcal N\times 1$ and the matrix size of $C$
is  $1\times \mathcal N.$

Similarly, the matrix triplet $(\bar A,\bar B,\bar C)$ are chosen as
\begin{equation}\label{4.6}
\bar A=\begin{bmatrix}
\bar A_1&0&\cdots&0&0\\
0&\bar A_2&\cdots&0&0\\
\vdots&\vdots&\ddots&\vdots&\vdots\\
0&0&\cdots&\bar A_{\bar N-1}&0\\
0&0&\cdots&0&\bar A_{\bar N}
\end{bmatrix}, \quad \bar B=\begin{bmatrix}
\bar B_1\\
\bar B_2\\
\vdots\\
\bar B_{\bar N}
\end{bmatrix}, \quad 
\bar C=\begin{bmatrix}
\bar C_1&\bar C_2&\cdots&\bar C_{\bar N}
\end{bmatrix},
\end{equation}
where the matrix 
subtriplet $(\bar A_j,\bar B_j,\bar C_j)$ is given by
\begin{equation}\label{4.7}
\bar A_j:=\begin{bmatrix}
\bar\lambda_j&1&0&\cdots&0&0\\
0&\bar\lambda_j&1&\cdots&0&0\\
0&0&\bar\lambda_j&\cdots&0&0\\
\vdots&\vdots&\vdots&\ddots&\vdots&\vdots\\
0&0&0&\cdots&\bar\lambda_j&1\\
0&0&0&\dots&0&\bar\lambda_j
\end{bmatrix}, \quad
\bar B_j:=\begin{bmatrix}
0\\ \vdots \\
0\\
1
\end{bmatrix},
\end{equation}
\begin{equation}\label{4.8}
\bar C_j:=\begin{bmatrix}
\bar c_{j(\bar m_j-1)}&\bar c_{j(\bar m_j-2)}&\cdots&\bar c_{j1}&\bar c_{j0}
\end{bmatrix}.
\end{equation}
As seen from \eqref{4.7},
the $\bar m_j\times \bar m_j$ matrix $\bar A_j$ is in the Jordan canonical form with $\bar \lambda_j$ in its diagonal entries and the
$\bar m_j\times 1$ matrix
$\bar B_j$ has the scalar $0$ in the first  $\bar m_j-1$ entries and $1$ in the  $\bar m_j$th entry.
As seen from \eqref{4.8}, the entries of the $1\times \bar m_j$ matrix $\bar C_j$ contain
the bound-state normalization constants
$\bar c_{jk}$ for $0\le k\le \bar m_j-1.$
Analogous to \eqref{4.5}, we introduce the nonnegative integer $\bar{\mathcal N}$ as
\begin{equation*}
 \bar{\mathcal N}:=\displaystyle\sum_{j=1}^{\bar N} \bar m_j.
\end{equation*}
Then, from \eqref{4.6} it follows that the matrix size of $\bar A$
is $\bar{\mathcal N}\times  \bar{\mathcal N},$ 
the matrix size of $\bar B$ is
$\bar{\mathcal N}\times 1,$ and the matrix size of $\bar C$ is
$1\times  \bar{\mathcal N}.$

\section{The Marchenko method}
\label{section5}

In this section we introduce 
the Marchenko method for \eqref{1.1}. This is done by
deriving the Marchenko system of linear integral equations for \eqref{1.1}
using as input the two right reflection coefficients $R(\zeta)$ and $\bar R(\zeta)$
and the bound-state information consisting of the two matrix
triplets $(A,B,C)$ and $(\bar A,\bar B,\bar C).$ We then show how the 
potential
pair $(q,r)$ in \eqref{1.1} and the Jost solutions $\psi(\zeta,x)$ and
$\bar\psi(\zeta,x)$ to \eqref{1.1} are obtained from the solution of the Marchenko system.

For clarity, we first derive our Marchenko system in the absence of bound states. Then, 
we show how the bound states affect
the kernel and the nonhomogeneous term
in the Marchenko system and thus obtain the Marchenko system in the presence of bound states.
The next theorem presents the Marchenko system of integral equations for \eqref{1.1} in the absence of bound states.

\begin{theorem}
\label{theorem5.1} 
Assume that the potentials $q$ and $r$ in \eqref{1.1} belong to the Schwartz class $\mathcal S(\mathbb R),$ 
and suppose that  there are no bound states for \eqref{1.1}. Then, the Marchenko system of linear integral equations for
\eqref{1.1} is given by
\begin{equation}\label{5.1}
\begin{split}
\begin{bmatrix}
0&0\\ 
\noalign{\medskip}
0&0
\end{bmatrix}=&\begin{bmatrix}
\bar K_1(x,y)&K_1(x,y)\\ \noalign{\medskip}\bar K_2(x,y)&K_2(x,y)
\end{bmatrix}+ \begin{bmatrix}
0&\hat{\bar R}(x+y)\\ 
\noalign{\medskip}
\hat R(x+y)&0
\end{bmatrix}\\
\noalign{\medskip}
&+\displaystyle\int_x^\infty dz \begin{bmatrix}
-i\, K_1(x,z)\,\hat R'(z+y)&\bar K_1(x,z)\,\hat{\bar R}(z+y)\\ 
\noalign{\medskip}
K_2(x,z)\,\hat R(z+y)&i\,\bar K_2(x,z)\,\hat{\bar R}'(z+y)
\end{bmatrix},\qquad x<y.
\end{split}
\end{equation}
The quantities $\hat R(y)$ and $\hat{\bar R}(y)$ in \eqref{5.1} are related to the reflection coefficients $R(\zeta)$ and $\bar R(\zeta),$ respectively,
via the Fourier transforms given by
\begin{equation}\label{5.2}
\hat R(y):=\displaystyle\frac{1}{2\pi}\displaystyle\int_{-\infty}^\infty  
d\lambda\,\displaystyle\frac{R(\zeta)}{\zeta}\,e^{i\lambda y},\quad \hat{\bar R}(y):=\displaystyle\frac{1}{2\pi}
\displaystyle\int_{-\infty}^\infty  d\lambda\,\displaystyle\frac{\bar R(\zeta)}{\zeta}\,e^{-i\lambda y},
\end{equation}
with $\hat R'(y)$ and $\hat{\bar R}'(y)$ denoting the derivatives of $\hat R(y)$ and $\hat{\bar R}(y),$ respectively. We recall that 
$\lambda$ appearing in \eqref{5.2} is
related to $\zeta$ as in \eqref{1.13}. The four
quantities $K_1(x,y),$ $K_2(x,y),$ $\bar K_1(x,y),$ $\bar K_2(x,y)$ in \eqref{5.1} are related to
the components of the Jost solutions $\psi(\zeta,x)$ and $\bar \psi(\zeta,x)$ appearing in \eqref{3.2} as
\begin{equation}\label{5.3}
K_1(x,y):= 
\displaystyle\frac{1}{2\pi }\int_{-\infty}^\infty d\lambda \left[\displaystyle\frac{\psi_1(\zeta,x)}{\zeta}\right] e^{-i\lambda y},
\end{equation}
\begin{equation}\label{5.4}
K_2(x,y):= 
\displaystyle\frac{1}{2\pi }\int_{-\infty}^\infty d\lambda \left[\psi_2(\zeta,x)-e^{i\lambda x}\right] e^{-i\lambda y},
\end{equation}
\begin{equation}\label{5.5}
\bar K_1(x,y):= 
\displaystyle\frac{1}{2\pi }\int_{-\infty}^\infty 
d\lambda \left[\bar\psi_1(\zeta,x)-e^{-i\lambda x}\right] e^{i\lambda y},
\end{equation}
\begin{equation}\label{5.6}
\bar K_2(x,y):= 
\displaystyle\frac{1}{2\pi }\int_{-\infty}^\infty d\lambda
\left[\displaystyle\frac{\bar\psi_2(\zeta,x)}{\zeta}\right] e^{i\lambda y}.
\end{equation}

\end{theorem}

\begin{proof} 
For notational simplicity, we suppress the arguments and write
$\psi,$ $\bar\psi,$ $\phi,$ $\bar\phi,$ $T,$ $\bar T,$ $R,$ $\bar R$
for $\psi(\zeta,x),$ $\bar\psi(\zeta,x),$ $\phi(\zeta,x),$ 
$\bar\phi(\zeta,x),$ $T(\zeta),$ $\bar T(\zeta),$ $R(\zeta),$ $\bar R(\zeta),$ respectively. 
As seen from the asymptotics in \eqref{3.4} and \eqref{3.5},
the column-vector Jost solutions $\psi$ and $\bar\psi$ to \eqref{1.1} are linearly independent,
and hence they form a fundamental set of column-vector solutions to \eqref{1.1}.
Thus, each of the 
two column-vector Jost solutions $\phi$ and $\bar\phi$ can be expressed as a linear combinations of
$\psi$ and $\bar\psi.$ With the help of \eqref{3.2}--\eqref{3.11}, for $\zeta\in\mathbb R$ we obtain
\begin{equation}
\label{5.7}
\begin{cases}
\phi=\displaystyle\frac{1}{T}\,\bar\psi+\displaystyle\frac{R}{T}\,\psi,
\\ \noalign{\medskip}
\bar\phi=\displaystyle\frac{\bar R}{\bar T}\,\bar\psi+\displaystyle\frac{1}{\bar T}\,\psi.
\end{cases}
\end{equation}
We write \eqref{5.7} equivalently as
\begin{equation}\label{5.8}
\begin{cases}
T\,\phi=\bar\psi+R\,\psi,
\\ \noalign{\medskip}
\bar T\,\bar\phi=\bar R\,\bar\psi+\psi.
\end{cases}
\end{equation}
We would like to transform \eqref{5.8} to an equivalent form so that
it yields a Riemann--Hilbert problem in the complex $\lambda$-plane
separated into two regions by the real $\lambda$-axis. For this we proceed as follows.
From the two column-vector equations in \eqref{5.8}, we get
the $2\times 2$ matrix-valued system 
\begin{equation}\label{5.9}
\begin{bmatrix}
T\,\phi&\bar T\,\bar\phi
\end{bmatrix}=\begin{bmatrix}
\bar\psi&\psi
\end{bmatrix}+\begin{bmatrix}
R\,\psi&\bar R\,\bar\psi
\end{bmatrix}.
\end{equation}
With the help of \eqref{3.2} and \eqref{3.3}, we write \eqref{5.9} in terms of the components of the Jost solutions as
\begin{equation}\label{5.10}
\begin{bmatrix}
T\,\phi_1&\bar T\,\bar\phi_1\\
\noalign{\medskip}
T\,\phi_2&\bar T\,\bar\phi_2
\end{bmatrix}=\begin{bmatrix}
\bar\psi_1&\psi_1
\\ \noalign{\medskip}
\bar\psi_2&\psi_2
\end{bmatrix}+\begin{bmatrix}
R\,\psi_1&\bar R\,\bar\psi_1
\\ \noalign{\medskip}
R\,\psi_2&\bar R\,\bar\psi_2
\end{bmatrix}.
\end{equation}
We subtract the $2\times 2$ diagonal matrix $\begin{bmatrix}
e^{-i\lambda x}&0\\
0&e^{i\lambda x}
\end{bmatrix}$ from each side of \eqref{5.10}.  Then, we divide the off-diagonal entries by $\zeta$ in the resulting matrix equality.
This yields
\begin{equation}\label{5.11}
\begin{split}
\begin{bmatrix}
T\,\phi_1-e^{-i\lambda x}&\displaystyle\frac{1}{\zeta}\,\bar T\,\bar\phi_1\\
\noalign{\medskip}
\displaystyle\frac{1}{\zeta}\,T\,\phi_2&\bar T\,\bar\phi_2-e^{i\lambda x}
\end{bmatrix}=\begin{bmatrix}
\bar\psi_1-e^{-i\lambda x}&\displaystyle\frac{1}{\zeta}\,\psi_1
\\ \noalign{\medskip}
\displaystyle\frac{1}{\zeta}\,\bar\psi_2&\psi_2-e^{i\lambda x}
\end{bmatrix}
+\begin{bmatrix}
R\,\psi_1&\displaystyle\frac{1}{\zeta}\,\bar R\,\bar\psi_1
\\ \noalign{\medskip}
\displaystyle\frac{1}{\zeta}\,R\,\psi_2&\bar R\,\bar\psi_2
\end{bmatrix},\qquad \lambda\in\mathbb R.
\end{split}
\end{equation}
From Theorems~\ref{theorem3.4} and \ref{theorem3.6}
it follows that each entry in \eqref{5.11} is even in
$\zeta$ and hence is a function of
$\lambda.$ The matrix equality in \eqref{5.11} is the formulation of the
Riemann--Hilbert problem of determining the Jost solutions
$\psi$ and $\bar\psi$ when we use as input the two
reflection coefficients
$R(\zeta)$ and $\bar R(\zeta)$ in the absence of bound states for \eqref{1.1}.
Next, we apply the Fourier transform on \eqref{5.11} by using
$\int_{-\infty}^\infty d\lambda\,e^{i\lambda y}/2\pi$ in the first columns and 
by using $\int_{-\infty}^\infty d\lambda\,e^{-i\lambda y}/2\pi$ in the second columns. We then get
the $2\times 2$ matrix-valued equality given by
\begin{equation}
\label{5.12}
\text{\rm{LHS}}=\mathcal K(x,y)+\text{\rm{RHS}},
\end{equation}
where the $2\times 2$ matrix $\mathcal K(x,y)$ is defined as
\begin{equation}
\label{5.13}
\mathcal K(x,y):=\begin{bmatrix}
\bar K_1(x,y)&K_1(x,y)
\\
\noalign{\medskip}
\bar K_2(x,y)&K_2(x,y)
\end{bmatrix},
\end{equation}
with the entries $K_1(x,y),$ $K_2(x,y),$ $\bar K_1(x,y),$ and $\bar K_2(x,y)$ are as in
\eqref{5.3}--\eqref{5.6}, respectively. The $2\times 2$ matrix-valued quantities
$\text{\rm{LHS}}$
and
$\text{\rm{RHS}}$
appearing in \eqref{5.12} are defined, respectively, as
\begin{equation}
\label{5.14}
\text{\rm{LHS}}:=\begin{bmatrix}
\displaystyle\int_{-\infty}^\infty \frac{d\lambda}{2\pi}\left[T(\zeta)\,\phi_1(\zeta,x)-e^{-i\lambda x}\right]e^{i\lambda y}
&\displaystyle\int_{-\infty}^\infty \frac{d\lambda}{2\pi}\left[\displaystyle\frac{1}{\zeta }\,\bar T(\zeta)\,\bar\phi_1(\zeta,x)\right]e^{-i\lambda y}
\\
\noalign{\medskip}
\displaystyle\int_{-\infty}^\infty \displaystyle\frac{d\lambda}{2\pi}\left[\displaystyle\frac{1}{\zeta}\,T(\zeta)\,\phi_2(\zeta,x)\right]e^{i\lambda y}&\displaystyle\int_{-\infty}^\infty \frac{d\lambda}{2\pi}\left[\bar T(\zeta)\,\bar\phi_2(\zeta,x)-e^{i\lambda x}\right]e^{-i\lambda y}
\end{bmatrix},
\end{equation}
\begin{equation}
\label{5.15}
\text{\rm{RHS}}:=\begin{bmatrix}
\displaystyle\int_{-\infty}^\infty \frac{d\lambda}{2\pi}\left[R(\zeta)\,\psi_1(\zeta,x)\right]
e^{i\lambda y}&\displaystyle\int_{-\infty}^\infty \displaystyle\frac{d\lambda}{2\pi}\left[\displaystyle\frac{1}{\zeta }
\,\bar R(\zeta)\,\bar\psi_1(\zeta,x)\right]e^{-i\lambda y}
\\
\noalign{\medskip}
\displaystyle\int_{-\infty}^\infty \displaystyle\frac{d\lambda}{2\pi}\left[\displaystyle\frac{1}{\zeta}
\,R(\zeta)\,\psi_2(\zeta,x)\right]e^{i\lambda y}&\displaystyle\int_{-\infty}^\infty 
\frac{d\lambda}{2\pi}\left[\bar R(\zeta)\,\bar\psi_2(\zeta,x)\right]e^{-i\lambda y}
\end{bmatrix}.
\end{equation}
In the absence of bound states, using the continuity and analyticity of the Jost solutions stated in Theorem~\ref{theorem3.4} and the  large $\zeta$-asymptotics of the Jost solutions stated in Theorem~\ref{theorem3.5},  we conclude the following. When $x<y,$ the integrands in \eqref{5.3} and \eqref{5.4} are
analytic in $\lambda\in\mathbb C^+,$ are continuous in $\lambda\in\overline{\mathbb C^+},$ and 
vanish as $e^{i\lambda(y-x)}O(1/\lambda)$ as $\lambda\to\infty$ in
$\overline{\mathbb C^+}.$ Similarly, when $x<y,$ the integrands in \eqref{5.5} and \eqref{5.6} are
analytic in $\lambda\in\mathbb C^-,$ are continuous in $\lambda\in\overline{\mathbb C^-},$ and 
vanish as $e^{-i\lambda(y-x)}O(1/\lambda)$ as $\lambda\to\infty$ in
$\overline{\mathbb C^-}.$
Consequently, with the help of Jordan's lemma we conclude that the matrix $\mathcal K(x,y)$ in \eqref{5.13} is equal to 
the $2\times 2$ zero matrix when $x>y.$
On the other hand, in the absence of bound states, using the continuity and analyticity of the Jost solutions stated in Theorem~\ref{theorem3.4}, the large $\zeta$-asymptotics of the Jost solutions stated in Theorem~\ref{theorem3.5}, and the continuity and asymptotic properties of the scattering coefficients presented in 
Theorems~\ref{theorem3.6}--\ref{theorem3.8}, when $x<y$ 
we observe that the integrands in the $(1,1)$ and $(2,1)$ entries on
the right-hand side of \eqref{5.14} are analytic
in $\lambda\in\mathbb C^+,$ are continuous in $\lambda\in\overline{\mathbb C^+},$
and behave uniformly as $O(1/\lambda)$ as $\lambda\to\infty$ in $\overline{\mathbb C^+}.$
Similarly, when $x<y,$ with the help of
Theorems~\ref{theorem3.4}--\ref{theorem3.8}, we conclude that the integrands
in the $(1,2)$ and $(2,2)$ entries of \eqref{5.14} are analytic
in $\lambda\in\mathbb C^-,$ continuous in $\lambda\in\overline{\mathbb C^-},$
and decay as $O(1/\lambda)$ as $\lambda\to\infty$ in $\overline{\mathbb C^-}.$
Thus, when $x<y$ the matrix
$\text{\rm{LHS}}$ in \eqref{5.14} is equal to the $2\times 2$ zero matrix.
Moreover, with the help of Theorems~\ref{theorem3.4}-- \ref{theorem3.8}, 
we observe that each integrand in \eqref{5.3}--\eqref{5.6}, \eqref{5.14}, and \eqref{5.15} is continuous 
in $\lambda\in\mathbb R$ and 
decays as $O(1/\lambda)$ as 
$\lambda\to\pm\infty.$ Hence, the $L^2$-Fourier transforms in \eqref{5.3}--\eqref{5.6}, \eqref{5.14}, 
and \eqref{5.15} exist. From \eqref{5.3}--\eqref{5.6}, by using the inverse Fourier transform we obtain
\begin{equation}\label{5.16}
\displaystyle\frac{1}{\zeta }\,\psi_1(\zeta,x)=\displaystyle\int_x^\infty dy\,K_1(x,y)\,e^{i\lambda y},
\end{equation}
\begin{equation}\label{5.17}
\psi_2(\zeta,x)=e^{i\lambda x}+\displaystyle\int_x^\infty dy\,K_2(x,y)\,e^{i\lambda y},
\end{equation}
\begin{equation}\label{5.18}
\bar\psi_1(\zeta,x)=e^{-i\lambda x}+\displaystyle\int_x^\infty dy\,\bar K_1(x,y)\,e^{-i\lambda y},
\end{equation}
\begin{equation}\label{5.19}
\displaystyle\frac{1}{\zeta}\,\bar\psi_2(\zeta,x)=\displaystyle\int_x^\infty dy\,\bar K_2(x,y)\,e^{-i\lambda y},
\end{equation}
and from \eqref{5.2} we have
\begin{equation}\label{5.20}
\displaystyle\frac{R(\zeta)}{\zeta}=\displaystyle\int_{-\infty}^\infty ds\,\hat R(s)\,e^{-i\lambda s},
\quad 
\frac{\bar R(\zeta)}{\zeta}=\displaystyle\int_{-\infty}^\infty ds\,\hat{\bar R}(s)\,e^{i\lambda s}.
\end{equation}
By using the $y$-derivatives in \eqref{5.2}, we get
\begin{equation}\label{5.21}
\hat R'(y)=\displaystyle\frac{i}{2\pi }\displaystyle\int_{-\infty}^\infty  
d\lambda\,\displaystyle\frac{R(\zeta)}{\zeta}\lambda\,e^{i\lambda y},
\quad \hat{\bar R}'(y)=-\displaystyle\frac{i}{2\pi }\displaystyle\int_{-\infty}^\infty  
d\lambda\,\displaystyle\frac{\bar R(\zeta)}{\zeta}\lambda\,e^{-i\lambda y}.
\end{equation}
The use of the inverse Fourier transform on \eqref{5.21} yields
\begin{equation}\label{5.22}
\displaystyle\frac{R(\zeta)}{\zeta}\,\lambda=-i\displaystyle\int_{-\infty}^\infty ds\,\hat R'(s)\,e^{-i\lambda s},
\quad 
\frac{\bar R(\zeta)}{\zeta}\,\lambda=i\displaystyle\int_{-\infty}^\infty ds\,\hat{\bar R}'(s)\,e^{i\lambda s}.
\end{equation}
Next, we apply the Fourier transform on each component of the $2\times 2$ matrix $\text{\rm{RHS}}$ in 
\eqref{5.15}. For this, we proceed as follows. The $(1,1)$ entry on
the right-hand side of \eqref{5.15} can be equivalently expressed as 
\begin{equation}
\label{5.23}
\displaystyle\int_{-\infty}^\infty \frac{d\lambda}{2\pi}\left[R(\zeta)\,\psi_1(\zeta,x)\right]
e^{i\lambda y}=
\displaystyle\int_{-\infty}^\infty \frac{d\lambda}{2\pi}\,e^{i\lambda y}\left(\displaystyle\frac{1}{\zeta }\psi_1(\zeta,x)\right)\left(\frac{R(\zeta)}{\zeta}\,\lambda\right).
\end{equation}
We then use \eqref{5.16} and the first equality of \eqref{5.22}
on the right-hand side of \eqref{5.23}. This yields 
\begin{equation}
\label{5.24}
\displaystyle\int_{-\infty}^\infty \frac{d\lambda}{2\pi}\left[R(\zeta)\,\psi_1(\zeta,x)\right]
e^{i\lambda y}=-i\displaystyle\int_x^\infty  dz\,K_1(x,z)\,\hat R'(z+y),
\end{equation}
where we recall that $K_1(x,z)$ vanishes when $x>z.$ In a similar manner, the $(2,2)$ entry on
the right-hand side of \eqref{5.15} 
can equivalently be written as
\begin{equation}
\label{5.25}
\displaystyle\int_{-\infty}^\infty \frac{d\lambda}{2\pi}\left[\bar R(\zeta)\,\bar\psi_2(\zeta,x)\right]e^{-i\lambda y}=\displaystyle\int_{-\infty}^\infty \frac{d\lambda}{2\pi}
\,e^{-i\lambda y}\left(\displaystyle\frac{\bar\psi_2(\zeta,x)}{\zeta}\right)\left(\frac{\bar R(\zeta)}{\zeta}\,\lambda\right).
\end{equation}
We then use \eqref{5.19} and the second equality of \eqref{5.22}
on the right-hand side of \eqref{5.25}, and we get
\begin{equation}\label{5.26}
\displaystyle\int_{-\infty}^\infty \frac{d\lambda}{2\pi}\left[\bar R(\zeta)\,\bar\psi_2(\zeta,x)\right]e^{-i\lambda y}
=i\displaystyle\int_x^\infty  dz\,\bar K_2(x,z)\,\hat{\bar R}'(z+y),
\end{equation}
where we recall that $\bar K_2(x,z)$ vanishes when $x>z.$
Similarly, we use \eqref{5.17}, \eqref{5.18}, \eqref{5.20}, and we write
the $(1,2)$ entry and
the $(2,1)$ entry on the right-hand side of \eqref{5.15}, respectively, as
\begin{equation}\label{5.27}
\displaystyle\int_{-\infty}^\infty \displaystyle\frac{d\lambda}{2\pi}\left[\displaystyle\frac{1}{\zeta }\,
\bar R(\zeta)\,\bar\psi_1(\zeta,x)\right]e^{-i\lambda y}
=\hat{\bar R}(x+y)+\displaystyle\int_x^\infty  dz\,\bar K_1(x,z)\,\hat{\bar R}(z+y),
\end{equation}
\begin{equation}\label{5.28}
\displaystyle\int_{-\infty}^\infty \displaystyle\frac{d\lambda}{2\pi}\left[\displaystyle\frac{1}{\zeta}\,R(\zeta)\,
\psi_2(\zeta,x)\right]e^{i\lambda y}=\hat R(x+y)+\displaystyle\int_x^\infty  dz\,K_2(x,z)\,\hat R(z+y).
\end{equation}
Thus, by using \eqref{5.24}, \eqref{5.26}, \eqref{5.27}, and \eqref{5.28}
in \eqref{5.15}, we conclude that the
$2\times 2$ matrix $\text{\rm{RHS}}$ in \eqref{5.15} is
equal to the sum of the second and third terms on the right-hand side of \eqref{5.1}.
Hence, the proof of the theorem is complete.
\end{proof}

In the presence of bound states for \eqref{1.1},
we modify the proof of Theorem~\ref{theorem5.1} as follows.
In that case, the quantity $\text{\rm{LHS}}$ in \eqref{5.14} is no longer equal to the $2\times 2$ zero matrix
because we must take into
consideration the bound-state poles of the transmission coefficients $T(\zeta)$ and $\bar T(\zeta)$ in 
the evaluation of the integrals in the entries of $\text{\rm{LHS}}.$
Using the poles of $T(\zeta)$ and $\bar T(\zeta)$ and
the bound-state dependency constants for \eqref{1.1}, we evaluate the aforementioned integrals
explicitly. We then explicitly express the resulting integrals in terms of 
$(A,B,C)$ and $(\bar A,\bar B,\bar C)$ appearing in \eqref{4.2} and \eqref{4.8},
respectively. This yields the Marchenko system of integral 
equations when \eqref{1.1} has bound states.

For the Marchenko system of integral equations for \eqref{1.1} in the presence of bound states, we introduce the $2\times 2$ matrix-valued
quantities $\Omega(y)$ and $\bar\Omega(y)$ by letting
\begin{equation}\label{5.29}
\Omega(y):=\hat R(y)+C\,e^{iAy}B,\quad \bar\Omega(y):=\hat{\bar R}(y)+\bar C\,e^{-i\bar A y}\bar B,
\end{equation}
where $e^{iAy}$ and $e^{-i\bar A y}$ denote the respective matrix exponentials. Using the $y$-derivative 
in \eqref{5.29}, we get
\begin{equation}
\label{5.30}
\Omega'(y)=\hat R'(y)+i \,C A\, e^{iAy} B,\quad \bar\Omega'(y)=\hat{\bar R}'(y)-i\, \bar C \bar A \,e^{-i\bar A y} \bar B.
\end{equation}
We present the corresponding Marchenko system for \eqref{1.1} in
the next theorem without a proof because
that proof is similar to the proof of Theorem~4.2 of \cite{AEU2023b}, where
the Marchenko system for \eqref{1.7} is derived in the presence of bound states for \eqref{1.7}.

\begin{theorem}
\label{theorem5.2}
Suppose that $(q,r)$ in \eqref{1.1} belongs to the Schwartz class $\mathcal S(\mathbb R).$
Let $(A,B,C)$ and $(\bar A,\bar B,\bar C)$ be the matrix triplet 
pair representing the bound-state information
for \eqref{1.1}. Furthermore, let
$\Omega(y),$ $\bar\Omega(y),$ $\Omega'(y),$ $\bar\Omega'(y)$ be the quantities
in \eqref{5.29} and \eqref{5.30}, and let $K_1(x,y),$ $K_2(x,y),$ $\bar K_1(x,y),$ $\bar K_2(x,y)$ be the quantities in
\eqref{5.3}--\eqref{5.6}, respectively. Then, in the presence of bound states, the Marchenko system
of integral equations for \eqref{1.1} is given by
\begin{equation}\label{5.31}
\begin{split}
\begin{bmatrix}
0&0\\ \noalign{\medskip}0&0
\end{bmatrix}=&\begin{bmatrix}
\bar K_1(x,y)&K_1(x,y)\\ \noalign{\medskip}\bar K_2(x,y)&K_2(x,y)
\end{bmatrix}+ \begin{bmatrix}
0&\bar\Omega(x+y)\\ \noalign{\medskip}\Omega(x+y)&0
\end{bmatrix}\\
\noalign{\medskip}
&+\displaystyle\int_x^\infty dz\begin{bmatrix}
-i\,K_1(x,z)\,\Omega'(z+y)&\bar K_1(x,z)\,\bar\Omega(z+y)\\ \noalign{\medskip}
K_2(x,z)\,\Omega(z+y)&i\,\bar K_2(x,z)\,\bar\Omega'(z+y)
\end{bmatrix},\qquad x<y.
\end{split}
\end{equation}

\end{theorem}

The $2\times 2$ matrix-valued Marchenko system
in \eqref{5.31} is equivalent to the coupled system of four scalar-valued
integral equations 
given by
\begin{equation}\label{5.32}
\begin{cases}
\bar K_1(x,y)-i\displaystyle\int_x^\infty dz\,K_1(x,z)\,\Omega'(z+y)=0,
\qquad x<y,\\
\noalign{\medskip}
K_1(x,y)+\bar\Omega(x+y)+\displaystyle\int_x^\infty dz\,\bar K_1(x,z)\,\bar\Omega(z+y)=0,
\qquad x<y,\\
\noalign{\medskip}
\bar K_2(x,y)+\Omega(x+y)+\displaystyle\int_x^\infty dz\, K_2(x,z)\,\Omega(z+y)=0,
\qquad x<y,\\
\noalign{\medskip}
K_2(x,y)+i\displaystyle\int_x^\infty dz\,\bar K_2(x,z)\,\bar\Omega'(z+y)=0,
\qquad x<y.
\end{cases}
\end{equation}
By uncoupling the coupled Marchenko system in \eqref{5.32}, we
get
\begin{equation}\label{5.33}
\begin{cases}
K_1(x,y)+\bar\Omega(x+y)+i\displaystyle\int_x^\infty dz\, K_1(x,z) \int_x^\infty 
ds\,\Omega'(z+s)\,\bar\Omega(s+y)=0,\qquad x<y,
\\
\noalign{\medskip}
\bar K_2(x,y)+\Omega(x+y)-i\displaystyle\int_x^\infty dz\,\bar K_2(x,z) \int_x^\infty 
ds\,\bar\Omega'(z+s)\,\Omega(s+y)=0,\qquad x<y,
\end{cases}
\end{equation}
\begin{equation}\label{5.34}
\begin{cases}
\bar K_1(x,y)=i\displaystyle\int_x^\infty dz\,K_1(x,z)\,\Omega'(z+y),\qquad x<y,
\\ \noalign{\medskip}
K_2(x,y)=-i\displaystyle\int_x^\infty dz\,\bar K_2(x,z)\,\bar\Omega'(z+y),\qquad x<y.
\end{cases}
\end{equation}
Thus, the solution of the Marchenko system can be achieved as follows.
We first solve the two uncoupled integral equations in \eqref{5.33} and obtain $K_1(x,y)$ and
$\bar K_2(x,y),$ respectively.
We then use those values in \eqref{5.34} and recover $\bar K_1(x,y)$ and $K_2(x,y).$

In the next theorem, we show how $(q,r)$ in \eqref{1.1} is recovered from
the solution 
$\mathcal K(x,y)$ to the Marchenko system \eqref{5.31}. We recall that
$\mathcal K(x,y)$ is 
related to the
quantities $K_1(x,y),$ $\bar K_1(x,y),$ $K_2(x,y),$ $\bar K_2(x,y)$ as in \eqref{5.13}.

\begin{theorem}\label{theorem5.3}
Assume that $(q,r)$ in \eqref{1.1} belongs to the Schwartz class $\mathcal S(\mathbb R).$
Let $ \mathcal K(x,y)$ be the solution of the Marchenko system
\eqref{5.31}, with the four components $K_1(x,y),$ $K_2(x,y),$ $\bar K_1(x,y),$ $\bar K_2(x,y)$ 
as in \eqref{5.13}. As $y\to x^+,$ those four components yield
\begin{equation}\label{5.35}
K_1(x,x)=-\displaystyle\frac{q(x)}{2},
\end{equation}
\begin{equation}\label{5.36}
K_2(x,x)=-\displaystyle\frac{i\,q(x)\,r(x)}{4}-\displaystyle\frac{i}{4}\int_x^\infty dy\,q(y)\,r'(y)-\displaystyle\frac{1}{8}\int_x^\infty dy\,q(y)^2\,r(y)^2,
\end{equation}
\begin{equation}\label{5.37}
\bar K_1(x,x)=-\displaystyle\frac{i}{4}\int_x^\infty dy\,q(y)\,r'(y)-\displaystyle\frac{1}{8}\int_x^\infty dy\,q(y)^2\,r(y)^2,
\end{equation}
\begin{equation}\label{5.38}
\bar K_2(x,x)=-\displaystyle\frac{r(x)}{2},
\end{equation}
where $K_1(x,x),$ $K_2(x,x),$ $\bar K_1(x,x),$ $\bar K_2(x,x)$ are used
to denote $K_1(x,x^+),$ $K_2(x,x^+),$ $\bar K_1(x,x^+),$ $\bar K_2(x,x^+),$ respectively.
\end{theorem}

\begin{proof} We write \eqref{5.16} in the equivalent form as 
\begin{equation}\label{5.39}
\displaystyle\frac{\psi_1(\zeta,x)}{\zeta}= 
\displaystyle\int_x^\infty dy\,K_1(x,y)\,\frac{d}{dy}\left(\frac{e^{i\lambda y}}{i\lambda}\right),
\end{equation}
where we recall that $\lambda$ is related to $\zeta$ as in \eqref{1.13}.
Using integration by parts in \eqref{5.39}, we get 
\begin{equation}
\label{5.40}
\displaystyle\frac{\psi_1(\zeta,x)}{\zeta}=
-K_1(x,x)\,\displaystyle\frac{e^{i\lambda x}}{i\lambda}-\int_x^\infty dy\,\frac{e^{i\lambda y}}
{i\lambda}\frac{\partial\,K_1(x,y)}{\partial y},
\end{equation}
where we have used $K_1(x,+\infty)=0.$ As $\lambda\to\infty$ in $\overline{\mathbb C^+},$ from \eqref{5.40} we obtain
\begin{equation}
\label{5.41}
\displaystyle\frac{\psi_1(\zeta,x)}{\zeta}=
-K_1(x,x)\displaystyle\frac{e^{i\lambda x}}{i\lambda}+O\left(\displaystyle\frac{1}{\lambda^2}\right).
\end{equation}
Using the large $\zeta$-asymptotics of $\psi_1(\zeta,x)$ given in \eqref{3.42} on the left-hand side of \eqref{5.41}, we get
\begin{equation}
\label{5.42}
e^{i\lambda x}\left[\displaystyle\frac{q(x)}{2i\lambda} +O\left(\displaystyle\frac{1}{\lambda^2}\right)\right]=
-\displaystyle\frac{K_1(x,x)\,e^{i\lambda x}}{i\lambda}+O\left(\displaystyle\frac{1}{\lambda^2}\right), \qquad \lambda\to\infty  \text{\rm{ in }} \overline{\mathbb C^+}.
\end{equation}
By comparing the $O(1/\lambda)$-terms on both sides of \eqref{5.42}, we establish \eqref{5.35}.
The equalities \eqref{5.36}--\eqref{5.38} are established in a similar manner.
\end{proof}

The next theorem shows how the quantity $E(x),$ the potential pair $(q,r),$ and the Jost solutions
$\psi(\zeta,x)$ and $\bar\psi(\zeta,x)$ to \eqref{1.1} are obtained from the solution $\mathcal K(x,y)$ to the Marchenko system \eqref{5.31}. 

\begin{theorem}
\label{theorem5.4}
Suppose that $(q,r)$ in \eqref{1.1} belongs to the Schwartz class $\mathcal S(\mathbb R).$
Let  $\mathcal K(x,y)$ be the solution of the Marchenko system \eqref{5.31} with the components
$K_1(x,y),$ $\bar K_1(x,y),$ $K_2(x,y),$ $\bar K_2(x,y)$ as in \eqref{5.13}.
The quantity $E(x)$ in \eqref{2.3}, the constant $\mu$ in \eqref{3.23}, the potential pair $(q,r),$
and the Jost solutions $\psi(\zeta,x)$ and $\bar\psi(\zeta,x)$ to 
\eqref{1.1} are obtained from $\mathcal K(x,y)$ as follows:

\begin{enumerate}

\item[\text{\rm(a)}] The quantities $E(x)$ and $\mu$ are recovered as
\begin{equation}\label{5.43}
E(x)=\exp\left(2i\displaystyle\int_{-\infty}^{x}dz\,P(z)\right), \quad \mu=4\displaystyle\int_{-\infty}^\infty dz\,P(z),
\end{equation}
where $P(x)$ is the scalar quantity constructed from $\bar K_1(x,y)$ and $K_2(x,y)$ as
\begin{equation}\label{5.44}
P(x):=K_1(x,x)\,\bar K_2(x,x).
\end{equation}

\item[\text{\rm(b)}] The potential pair $(q,r)$ is recovered via
\begin{equation}\label{5.45}
q(x)=-2K_1(x,x),
\end{equation}
\begin{equation}\label{5.46}
r(x)=-2\bar K_2(x,x).
\end{equation}

\item[\text{\rm(c)}] The Jost solutions $\psi(\zeta,x)$ and $\bar\psi(\zeta,x)$ 
to \eqref{1.1} are recovered via
\begin{equation}\label{5.47}
\psi_1(\zeta,x)= \zeta\displaystyle\int_x^\infty dy\,K_1(x,y)\,e^{i\zeta^2y},
\end{equation}
\begin{equation}\label{5.48}
\psi_2(\zeta,x)=e^{i\zeta^2x}+\displaystyle\int_x^\infty dy\,K_2(x,y)\,e^{i\zeta^2 y},
\end{equation}
\begin{equation}\label{5.49}
\bar\psi_1(\zeta,x)= e^{-i\zeta^2x}+\displaystyle\int_x^\infty dy\,\bar K_1(x,y)\,e^{-i\zeta^2 y},
\end{equation}
\begin{equation}\label{5.50}
\bar\psi_2(\zeta,x)= \zeta \displaystyle\int_x^\infty dy\,\bar K_2(x,y)\,e^{-i\zeta^2y},
\end{equation}
where we recall that $\psi_1(\zeta,x),$ $\psi_2(\zeta,x),$ $\bar\psi_1(\zeta,x),$ and $\bar\psi_2(\zeta,x)$ 
are the components of the Jost solutions as described in \eqref{3.2}.
	
\end{enumerate}
\end{theorem}

\begin{proof}
Using \eqref{5.35} and \eqref{5.38} on the right-hand side of \eqref{5.44}, we see
 that $P(x)$ defined in 
\eqref{5.44} is related to $(q,r)$ as
\begin{equation}
\label{5.51}
P(x)=\displaystyle\frac{q(x)\,r(x)}{4}.
\end{equation}
Thus, using \eqref{2.3}, \eqref{5.44}, and \eqref{5.51} we obtain the first equality of \eqref{5.43}. 
By using \eqref{5.51} in \eqref{3.23}, we get the second equality of \eqref{5.43}.
This completes the proof of (a). The proof of (b) directly follows from \eqref{5.35} and \eqref{5.38}.
The proof of (c) is obtained directly from \eqref{5.16}--\eqref{5.19}.
\end{proof}

\section{The reflectionless case}
\label{section6}

In the reflectionless case for \eqref{1.1}, the solution of the Marchenko system
\eqref{5.31} can be explicitly constructed in terms of the matrix triplet pair
$(A,B,C)$ and $(\bar A,\bar B,\bar C)$ appearing in \eqref{5.29}.
This can be seen as follows. When $R(\zeta)\equiv 0$ and $\bar R(\zeta)\equiv 0$
for \eqref{1.1}, from \eqref{5.2}, \eqref{5.29}, and \eqref{5.30} it follows that the kernel terms
appearing in \eqref{5.31} are explicitly expressed in terms of
the matrix triplet pairs as
\begin{equation}\label{6.1}
\Omega(y)=C\,e^{iAy}\,B,\quad \bar\Omega(y)=\bar C\,e^{-i\bar A y}\,\bar B,
\end{equation}
\begin{equation}\label{6.2}
\Omega'(y)=i\,C A \,e^{iAy}\,B,\quad \bar\Omega'(y)=-i\bar C\bar A \,e^{-i\bar A y}\,\bar B.
\end{equation}
As seen from \eqref{6.1} and \eqref{6.2}, the quantities
$\Omega(x+y),$ $\Omega'(x+y),$ $\bar\Omega(x+y),$ $\bar\Omega'(x+y)$
are separable in $x$ and $y,$ and hence
the linear integral system \eqref{5.31}, or the equivalent uncoupled system consisting of
\eqref{5.33} and \eqref{5.34}, can be explicitly solved by the
methods of linear algebra.
Since the potential pair $(q,r)$ in \eqref{1.1} and the Jost solutions
to \eqref{1.1} can be explicitly expressed in terms of the solution of
the Marchenko system \eqref{5.31}, we can in turn express the potential pair and the
Jost solutions explicitly in terms of 
$(A,B,C)$ and $(\bar A,\bar B,\bar C).$
Our formulas for the potential pair and the Jost solutions
are explicit and have compact forms because they contain matrix exponentials
involving $A$ and $\bar A.$ The advantage of
using matrix exponentials in our formulas is that those formulas are valid
whether our input data set contains a small number of simple bound states or 
any large number of bound states with any multiplicities.
The compact expressions for such formulas expressed with the help of
matrix exponentials can easily be converted to the corresponding expressions in terms of
elementary functions without using 
any matrix exponentials. The latter expressions become extremely lengthy
as the number and multiplicities of the bound states increase. Such
lengthy expressions can be displayed explicitly with the help of a symbolic 
computational software such as Mathematica.
In this section, we present our formulas compactly expressed in terms of
matrix exponentials in the reflectionless case.
We also illustrate the corresponding expressions in
elementary functions in two explicit examples without the use of matrix exponentials.

In the next theorem we present the explicit solution formula for the Marchenko system \eqref{5.31} corresponding to the reflectionless Marchenko kernels given in \eqref{6.1} and \eqref{6.2},
which are uniquely determined by 
$(A,B,C)$ and $(\bar A,\bar B,\bar C)$ with the help of matrix exponentials.
We know from Section~\ref{section4} that the eigenvalues
of $A$ are located in $\mathbb C^+$ and that
the eigenvalues
of $\bar A$ are located in $\mathbb C^-.$

\begin{theorem}
\label{theorem6.1}
When the reflectionless scattering data set in \eqref{6.1} and \eqref{6.2} is used as input,
the solution of the Marchenko system \eqref{5.31} is expressed in closed form as
\begin{equation}\label{6.3}
K_1(x,y)
=-\bar C\,e^{-i\bar A x}\,\bar\Gamma(x)^{-1}\,e^{-i\bar A y}\,\bar B,
\end{equation}
\begin{equation}\label{6.4}
K_2(x,y)
=C\,e^{iAx}\,\Gamma(x)^{-1}\,e^{iAx}\,M\,
\bar A\,e^{-i\bar A (x+y)}\,\bar B,
\end{equation}
\begin{equation}\label{6.5}
\bar K_1(x,y)
=\bar C\,e^{-i\bar A x}\,\bar\Gamma(x)^{-1}\,e^{-i\bar A x}\,
\bar M\,A\,e^{iA(x+y)}\,B,
\end{equation}
\begin{equation}\label{6.6}
\bar K_2(x,y)
=-C\,e^{iAx}\,\Gamma(x)^{-1}\,e^{iAy}\,B.
\end{equation}
Here, $(A,B,C)$ and $(\bar A,\bar B,\bar C)$
are the two matrix triplets appearing in \eqref{5.29} and \eqref{5.30} with
the eigenvalues
of $A$ located in $\mathbb C^+$ and 
the eigenvalues
of $\bar A$ located in $\mathbb C^-.$ 
The quantities $\Gamma(x),$ $\bar\Gamma(x),$ $M,$ and $\bar M$
are the matrices defined in terms of
the two matrix triplets as
\begin{equation}\label{6.7}
\Gamma(x):=I-e^{iAx}\,M\,\bar A\,e^{-2i\bar A x}\,\bar M\,e^{iAx},
\end{equation}
\begin{equation}\label{6.8}
\bar\Gamma(x):=I-e^{-i\bar A x}\,\bar M\,A\,e^{2iAx}\,M\,e^{-i\bar A x},
\end{equation}
\begin{equation}\label{6.9}
M:=\int_{0}^\infty dz\,e^{iAz}\,B\,\bar C\,e^{-i\bar A z},\quad 
\bar M:=\int_{0}^\infty dz\,e^{-i\bar A z}\,\bar B\,C\,e^{i A z},
\end{equation}
with $I$ denoting the identity matrix whose size is not necessarily the
same 
in different appearances.
\end{theorem}

\begin{proof}
Since the Marchenko system \eqref{5.31} is equivalent to the uncoupled
system given in \eqref{5.33} and \eqref{5.34}, we use \eqref{6.1} and \eqref{6.2} as input to
that uncoupled system. The first line of \eqref{5.33} yields
\begin{equation}\label{6.10}
K_1(x,y)
+\bar C\,e^{-i\bar Ax-i\bar Ay}\,\bar B
+i\displaystyle\int_x^\infty dz\int_x^\infty 
ds\,K_1(x,z)\,i\,C\,A\,e^{iAz+iAs}\,B\,\bar C\,e^{-i\bar As-i\bar A y}\,\bar B=0.
\end{equation}
Since the matrix products in the second and third terms on the left-hand side of
\eqref{6.10} contain $e^{-i\bar Ay}\bar B$ as their common last factors, the solution $K_1(x,y)$ has the form
\begin{equation}\label{6.11}
K_1(x,y)=H_1(x)\,e^{-i\bar A y}\,\bar B,
\end{equation}
where $H_1(x)$ is the matrix to be determined. Using \eqref{6.11} in \eqref{6.10} we get
\begin{equation}\label{6.12}
H_1(x)\left[I-\displaystyle\int_x^\infty dz\int_x^\infty 
ds\,e^{-i\bar A z}\,\bar B\,C\,e^{iAz}\,A\,e^{iAs}\,B\bar C\,e^{-i\bar A s}\right]
=-\bar C\,e^{-i\bar A x}.
\end{equation}
The matrix in the brackets in \eqref{6.12} is equal to $\bar\Gamma(x)$ defined 
in \eqref{6.8}, and this can be seen by observing that we
can write the first and second equalities in \eqref{6.9} in the respective equivalent forms as
\begin{equation}\label{6.13}
\displaystyle\int_x^\infty ds
\,e^{i A s}\,B\,\bar C\,e^{-i \bar A s}=e^{iAx} \,M\, e^{-i\bar A x},
\end{equation}
\begin{equation}\label{6.14}
\displaystyle\int_x^\infty dz
\,e^{-i\bar A z}\,\bar B\,C\,e^{iAz}=e^{-i\bar A x}\,\bar M\, e^{iAx}.
\end{equation}
Since the eigenvalues of $A$ are located in $\mathbb C^+$ and
the eigenvalues of $\bar A$ are in $\mathbb C^-,$ 
the two integrals in \eqref{6.9} are well defined. From \eqref{6.9} we also see that
the matrices $M$ and $\bar M$ can alternatively
be obtained from 
$(A,B,C)$ and $(\bar A,\bar B,\bar C)$
by solving the respective linear systems given by
\begin{equation*}
i M\bar A-iA M=B \bar C, \quad i\bar A \bar M-i\bar M  A=\bar B C.
\end{equation*}
Thus, from \eqref{6.12} we get
\begin{equation}\label{6.16}
H_1(x)=-\bar C\,e^{-i\bar A x} \,\bar \Gamma(x)^{-1}.\end{equation}
Finally, using \eqref{6.16} in \eqref{6.11} we see that \eqref{6.3} holds.
We establish \eqref{6.6} in a similar manner, by using 
\eqref{6.1} and \eqref{6.2} as input in the second line of \eqref{5.33}
and by utilizing \eqref{6.13} and \eqref{6.14}.
Next, by using \eqref{6.3} and the first equality in \eqref{6.2} as input
to the first line of \eqref{5.34}, we obtain
\begin{equation}\label{6.17}
\bar K_1(x,y)
=\displaystyle\int_x^\infty dz\,\bar C\,e^{-i\bar A x}\,\bar\Gamma(x)^{-1}\,e^{-i\bar A z}\,
\bar B\,C\,A\,e^{iA(z+y)}\,B.
\end{equation}
With the help of \eqref{6.14}, from \eqref{6.17}  we get \eqref{6.5}. Finaly, using \eqref{6.6} and the second equation in \eqref{6.2} as input
to the second line of \eqref{5.34}, we have
\begin{equation}\label{6.18}
K_2(x,y)
=\displaystyle\int_x^\infty dz\,C\,e^{iAx}\,\Gamma(x)^{-1}\,e^{iAz}\,B\,\bar C\,
\bar A\,e^{-i\bar A (z+y)}\,\bar B.
\end{equation}
Using \eqref{6.13} in \eqref{6.18} we obtain \eqref{6.4}. Thus, the proof is complete.
\end{proof}

In the next theorem we present the explicit expressions for the key quantity $E(x)$ in \eqref{2.3}, the
constant $\mu$ in \eqref{3.23}, the potentials $q$ and $r$ in \eqref{1.1}, and the Jost solutions $\psi(\zeta,x)$ and  $\bar \psi(\zeta,x)$ 
corresponding to the reflectionless scattering data set described by the 
matrix triplet pair $(A,B,C)$ and $(\bar A,\bar B,\bar C).$

\begin{theorem}
\label{theorem6.2}
Let the quantities $\Omega(y)$ and $\bar\Omega(y)$ appearing in \eqref{6.1} 
comprise the reflectionless input scattering data set
for \eqref{1.1} with the eigenvalues
of $A$ located in $\mathbb C^+$ and 
the eigenvalues
of $\bar A$ located in $\mathbb C^-.$  
We then have the following:

\begin{enumerate}

\item[\text{\rm(a)}]  The corresponding scalar quantity $E(x)$ defined in \eqref{2.3} and 
the corresponding constant $\mu$ defined in \eqref{3.23} 
are uniquely and explicitly determined
in terms of the matrix triplet pair $(A,B,C)$ and $(\bar A,\bar B,\bar C)$ as
\begin{equation}\label{6.19}
E(x)=\exp\left(2i\displaystyle\int_{-\infty}^x dzP(z)\right), \quad \mu=4\displaystyle\int_{-\infty}^\infty dz\,P(z),
\end{equation}
where $P(x)$ is the scalar-valued function of $x$ given by
\begin{equation}\label{6.20}
P(x)=\bar C\,e^{-i\bar A x}\,\bar\Gamma(x)^{-1}\,e^{-i\bar A x}\,\bar B\,C\,e^{iAx}\,\Gamma(x)^{-1}\,e^{iAx}\,B.
\end{equation}
We recall that the $2\times 2$ matrix-valued quantity $\Gamma(x)$ in \eqref{6.20} is explicitly determined by the matrix triplet
pair as described in \eqref{6.7}. Similarly, the $2\times 2$ matrix-valued quantity $\bar\Gamma(x)$ in \eqref{6.20} is explicitly determined by the matrix triplet
pair as described in \eqref{6.8}.

\item[\text{\rm(b)}]  The corresponding potentials $q$ and $r$ in \eqref{1.1}
are uniquely and explicitly determined
in terms of the matrix triplet pair as
\begin{equation}\label{6.21}
q(x)=2\,\bar C\,e^{-i\bar A x}\,\bar\Gamma(x)^{-1}\,e^{-i\bar A x}\,\bar B,
\end{equation}
\begin{equation}\label{6.22}
r(x)=2\,C\,e^{i A x}\,\Gamma(x)^{-1}\,e^{i A x}\,B.
\end{equation}

\item[\text{\rm(c)}]  The Jost solutions $\psi(\zeta,x)$ and $\bar\psi(\zeta,x)$
to \eqref{1.1} appearing in \eqref{3.4}
are explicitly expressed in terms of the
pair of matrix triplet pair as
\begin{equation}\label{6.23}
\psi(\zeta,x)=\begin{bmatrix} \zeta e^{i \zeta^2 x}
\,g_1(\zeta,x)\\
\noalign{\medskip}
 e^{i \zeta^2 x}
\,g_2(\zeta,x)\end{bmatrix},
\quad
\bar\psi(\zeta,x)=\begin{bmatrix} e^{-i \zeta^2 x}
\,g_3(\zeta,x)\\
\noalign{\medskip}
\zeta e^{-i \zeta^2 x}
\,g_4(\zeta,x)\end{bmatrix},
\end{equation}
where the quantities $g_1(\zeta,x),$
$g_2(\zeta,x),$ $g_3(\zeta,x),$ $g_4(\zeta,x)$
are defined as
\begin{equation}\label{6.24}
g_1(\zeta,x):=-
i\,\bar C\,e^{-i\bar A x}\,\bar\Gamma(x)^{-1}\,e^{-i\bar A x} (\bar A -\zeta^2 I)^{-1}\bar B,
\end{equation}
\begin{equation}\label{6.25}
g_2(\zeta,x):=1-i\,C\,e^{iAx}\,\Gamma(x)^{-1}\,e^{iAx}\,M\,\bar A\,e^{-2i\bar A x}\,(\bar A -\zeta^2 I)^{-1}\,\bar B,
\end{equation}
\begin{equation}\label{6.26}
g_3(\zeta,x):=1+i\,\bar C\,e^{-i\bar A x}\,\bar\Gamma(x)^{-1}\,e^{-i\bar A x}\,\bar M\,A\,e^{2iAx}\,(A-\zeta^2 I)^{-1}\,B,
\end{equation}
\begin{equation}\label{6.27}
g_4(\zeta,x):=
-i\,C\,e^{iAx}\,\Gamma(x)^{-1}\,e^{iAx}\,(A-\zeta^2 I)^{-1}\,B.
\end{equation}
We recall that the constant matrices $M$ and $\bar M$ 
appearing in \eqref{6.9} are uniquely determined by
$(A,B,C)$ and $(\bar A,\bar B,\bar C),$
and hence
the scalar quantities defined in 
\eqref{6.24}--\eqref{6.27} are each 
uniquely and explicitly determined in terms of that matrix triplet pair.
	
\item[\text{\rm(d)}] The transmission coefficients
$T(\zeta)$ and $\bar T(\zeta)$ corresponding to 
the reflectionless input scattering data set associated with \eqref{6.1} are explicitly determined by
the matrix triplet pair as
\begin{equation}\label{6.28}
T(\zeta)=\displaystyle\frac{1}{g_2(\zeta,-\infty)},
\end{equation}
\begin{equation}\label{6.29}
\bar T(\zeta)=\displaystyle\frac{1}{g_3(\zeta,-\infty)},
\end{equation}
where, as seen from \eqref{6.25} and \eqref{6.26},
the quantities $g_2(\zeta,-\infty)$ and $g_3(\zeta,-\infty)$ are explicitly determined
by the matrix triplet pair.

\end{enumerate}
\end{theorem}

\begin{proof} 
For the proof of (a), we proceed as follows. Letting $y=x^+$ in \eqref{6.3} and \eqref{6.6}, we obtain
\begin{equation}\label{6.30}
K_1(x,x)
=-\bar C\,e^{-i\bar A x}\,\bar\Gamma(x)^{-1}\,e^{-i\bar A x}\,\bar B,
\end{equation}
\begin{equation}\label{6.31}
\bar K_2(x,x)=-C\,e^{iAx}\,\Gamma(x)^{-1}\,e^{iAx}\,B.
\end{equation}
Using \eqref{6.30} and \eqref{6.31} in \eqref{5.44}, we obtain the scalar quantity $P(x)$ given in
\eqref{6.20}. Finally, using \eqref{6.20} in \eqref{5.43}, we obtain \eqref{6.19}. Hence, the proof of
(a) is complete. The proof of (b) is obtained by using
\eqref{6.30} and \eqref{6.31} in \eqref{5.45} and \eqref{5.46}, respectively.
For the proof of (c), we proceed as follows.
We use \eqref{6.3} on the right-hand side of \eqref{5.47}, and we get
\begin{equation}\label{6.32}
\psi_1(\zeta,x)=- \zeta\displaystyle\int_x^\infty dy\,\bar C\,e^{-i\bar A x}\,\bar\Gamma(x)^{-1}\,e^{-i\bar A y}\,\bar B\,e^{i\zeta^2y},
\end{equation}
where we recall that $\psi_1(\zeta,x)$ is the first component of $\psi(\zeta,x)$ as indicated in the first
equality of \eqref{3.2}.
We write \eqref{6.32} in the equivalent form as
\begin{equation}\label{6.33}
\psi_1(\zeta,x)=- \zeta\,\bar C\,e^{-i\bar A x}\,\bar\Gamma(x)^{-1}
\left(\displaystyle\int_x^\infty dy\,e^{-i(\bar A-\zeta^2I) y}\right)\bar B.
\end{equation}
By evaluating the integral term on the right-hand side of \eqref{6.33}, we
write
\eqref{6.33} in the equivalent form as
\begin{equation*}
\psi_1(\zeta,x)= -i\zeta  
\,\bar C\,e^{-i\bar Ax}\,\bar\Gamma(x)^{-1}\,e^{-i(\bar A-\zeta^2I)  x} (\bar A -\zeta^2 I)^{-1}\bar B.
\end{equation*}
Thus, the first component in the first equality in \eqref{6.23} is established.
In a similar manner, we establish the equality in the second component
in the first equality in \eqref{6.23} and also establish the second equality in \eqref{6.23}.
For this we use \eqref{6.4}--\eqref{6.6} on the 
right-hand sides of \eqref{5.48}--\eqref{5.50}, respectively, and we evaluate
the respective integral terms there explicitly. This completes the proof of (c). 
For the proof of (d), we proceed as follows.
From \eqref{3.8} and \eqref{3.9}, we have
\begin{equation}
\label{6.35}
\psi_2(\zeta,x)=\displaystyle\frac{1}{T(\zeta)}\,e^{i\zeta^2 x}\left[1+o(1)\right],\qquad x\to-\infty,
\end{equation}
\begin{equation}
\label{6.36}
\bar\psi_1(\zeta,x)=\displaystyle\frac{1}{\bar T(\zeta)}\,e^{-i\zeta^2 x}\left[1+o(1)\right],\qquad x\to-\infty,
\end{equation}
where we recall that $\psi_2(\zeta,x)$ is the second
component of $\psi(\zeta,x)$ and
$\bar\psi_1(\zeta,x)$ is the first
component of $\bar\psi(\zeta,x)$ as indicated in \eqref{3.2}.
Using \eqref{6.23} on the left-hand sides of \eqref{6.35} and \eqref{6.36}, respectively,
we obtain \eqref{6.28} and \eqref{6.29}. Thus, the proof of (d) is complete.
\end{proof}

As seen from \eqref{6.21} and \eqref{6.22}, 
in the reflectionless case, the potential pair $(q,r)$ in \eqref{1.1} is uniquely determined
when the reflectionless scattering data set is specified in terms of
the matrix triplet pair $(A,B,C)$ and $(\bar A,\bar B,\bar C).$
The formulas \eqref{6.21} and \eqref{6.22} contain matrix exponentials in case we have multiple bound states
or the bound states are not simple. This does not present any difficulty as the matrix exponentials
in \eqref{6.21} and \eqref{6.22} can easily be evaluated in terms of elementary functions. The same remark applies
also for the explicit evaluations of the Jost solutions to \eqref{1.1} in the reflectionless case. The Jost solutions
$\psi(\zeta,x)$ and $\bar\psi(\zeta,x)$ to \eqref{1.1} can be explicitly expressed in terms of elementary functions
by using \eqref{6.23} and by expressing matrix exponentials there 
in terms of elementary functions.

In the next example, we illustrate the use of Theorem~\ref{theorem6.2}
for the evaluation of the potential pair $(q,r)$, the Jost solutions
$\psi(\zeta,x)$ and $\bar\psi(\zeta,x),$ and
the transmission coefficients $T(\zeta)$ and $\bar T(\zeta)$ for \eqref{1.1} in the reflectionless case
by using the input scattering data set 
consisting of a matrix triplet pair corresponding to two simple bound states. 

\begin{example}
\label{example6.3}
\normalfont
In the reflectionless case, as the input scattering data set
we use the matrix triplets $(A,B,C)$ and $(\bar A,\bar B,\bar C)$ given by
\begin{equation}\label{6.37}
A=\begin{bmatrix}i\end{bmatrix},\quad B=\begin{bmatrix}1\end{bmatrix},\quad C=\begin{bmatrix}
2\end{bmatrix},
\quad \bar A =\begin{bmatrix}
-i\end{bmatrix},\quad \bar B=\begin{bmatrix}
1\end{bmatrix},\quad\bar C=\begin{bmatrix}
2\end{bmatrix}.
\end{equation}
From the expressions for $A$ and $\bar A$ in \eqref{6.37}, we see that 
the transmission coefficient
$T(\zeta)$
has a simple bound-state pole of 
at $\lambda=i$ and that the transmission coefficient
$\bar T(\zeta)$
has a simple bound-state pole
at $\lambda=-i,$ where we recall that $\lambda$ and $\zeta$ are related to each other as in \eqref{1.13}.
Using \eqref{6.37} in \eqref{6.21} and \eqref{6.22}, we obtain the corresponding potentials $q$ and $r$ as
\begin{equation}\label{6.38}
q(x)=\displaystyle\frac{4e^{2x}}{-i+e^{4x}}, \quad
r(x)=\displaystyle\frac{4e^{2x}}{i+e^{4x}}, \qquad x\in\mathbb R.
\end{equation}
From \eqref{6.38} we observe that the potentials $q(x)$ and $r(x)$ are related to each other as 
$r(x)=q(x)^\ast,$
where we recall that an asterisk denotes complex conjugation. From \eqref{6.38} we also see that
$q$ and $r$ are each complex valued, behave as $O(e^{-2|x|})$ as $x\to\pm\infty,$ and 
belong to the Schwartz class. Using \eqref{6.37} in \eqref{6.23}--\eqref{6.27}, we obtain the corresponding
Jost solutions $\psi(\zeta,x)$ and $\bar\psi(\zeta,x)$ as
\begin{equation}\label{6.39}
\psi(\zeta,x)=\begin{bmatrix}
\displaystyle\frac{2\zeta\,e^{2x+i\zeta^2x}  }{(\zeta^2+i)\left(1+i e^{4x}\right)}\\
\noalign{\medskip}
e^{i\zeta^2 x}\left(1+\displaystyle\frac{2}{(\zeta^2+i)\left(i+e^{4x}\right)}\right)
\end{bmatrix},
\end{equation}
\begin{equation}\label{6.40}
\bar\psi(\zeta,x)=\begin{bmatrix}
e^{-i\zeta^2 x}\left(1+\displaystyle\frac{2}{(\zeta^2-i)\left(-i+e^{4x}\right)}\right)\\
\noalign{\medskip}
\displaystyle\frac{2\zeta\,e^{2x-i\zeta^2x}  }{(\zeta^2-i)\left(1-i e^{4x}\right)}
\end{bmatrix}.
\end{equation}
Using
the asymptotics of \eqref{6.39} and \eqref{6.40} as $x\to-\infty$ and 
comparing those asymptotics with \eqref{6.35} and \eqref{6.36},
we obtain the transmission coefficients
as
\begin{equation}\label{6.41}
 T(\zeta)=\displaystyle\frac{\lambda+i}{\lambda-i},\quad \bar T(\zeta)=\displaystyle\frac{\lambda-i}{\lambda+i},
\end{equation}
where we again recall that $\lambda$ and $\zeta$ are related to each other as in \eqref{1.13}.
We note that the result in \eqref{6.41} can also be obtained by using \eqref{6.28} and \eqref{6.29}.
In this example, with the help of \eqref{6.7}--\eqref{6.9},
we evaluate the complex-valued scalar quantity $E(x)$ appearing in \eqref{2.3}.
Then, by using the second equality of \eqref{3.27}, we 
evaluate the constant $\mu$ in \eqref{3.23}.
We get
\begin{equation*}
E(x)=\exp \big(2i \tan^{-1}(e^{4x})\big), \quad\mu=2\pi,
\end{equation*}
where $\tan^{-1}$ is the single-valued branch of the real-valued
tangent inverse function taking values in the interval $(-\pi/2,\pi/2).$

\end{example}

In the next example, we illustrate the use of  \eqref{6.21} and \eqref{6.22}
to evaluate the potential pair $(q,r)$ in \eqref{1.1} in the reflectionless case
by using the input scattering data set 
consisting of a matrix triplet pair corresponding to four bound states, where
two of the bound states each have multiplicity two.

\begin{example}
\label{example6.4}
\normalfont
Consider the reflectionless scattering data with the
bound-state information described by the matrix triplets $(A,B,C)$ and $(\bar A,\bar B,\bar C)$ given by
\begin{equation}\label{6.43}
A=\begin{bmatrix}i&1&0\\
0&i&0\\
0&0&2i\end{bmatrix}, \quad B=\begin{bmatrix}0\\
1\\
1\end{bmatrix}, \quad
C=\begin{bmatrix}1&1&1\end{bmatrix},
\end{equation}
\begin{equation}\label{6.44}
\bar A=\begin{bmatrix}-i&1&0\\
0&-i&0\\
0&0&-2i\end{bmatrix}, \quad \bar B=\begin{bmatrix}0\\ 1
\\
1\end{bmatrix}, \quad
\bar C=\begin{bmatrix}1&1&1\end{bmatrix}.
\end{equation}
Using \eqref{6.43} and \eqref{6.44}  in \eqref{6.21} and \eqref{6.22}, after expressing all the matrix exponentials
in terms of elementary functions, we obtain the corresponding potentials $q$ and $r$ as
\begin{equation}\label{6.45}
q(x)=\displaystyle\frac{-48e^{2x}\left[22 -6i +12x+27 e^{2x} (\omega_1+\omega_2)\right]}{1 + 72 e^{4x}\left[72+ 812i+ 912i x-9(\omega_3+\omega_4)\right]}, \quad r(x)=q(x)^\ast, \qquad x\in\mathbb R,
\end{equation}
where we have defined
\begin{equation*}
\omega_1:=-i- 96i (8 + 3i + 6 x) e^{2x} +32 e^{4x} \left[79 - 42i + 12 x (11 + 6 x)\right],
\end{equation*}
\begin{equation*}
\omega_2:= 1296(1+ i + 2ix) e^{6x}+ 20736ie^{8x} +20736(i + 2x) e^{10x},
\end{equation*}
\begin{equation*}
\omega_3:= -32ix^2 +16(2 -3i + 2 x)e^{2x} +(2592-81i) e^{4x}-768(-9 + 5i + 6i x) e^{6x},
\end{equation*}
\begin{equation*}
\omega_4:=  2592(3 - 2i + 4x + 8x^2)e^{8x} +20736 i\,e^{12x}. 
\end{equation*}
From \eqref{6.45}, we see that $q$ and $r$ are each complex valued, 
belong to the Schwartz class, and
behave as $O(e^{-2|x|})$ as $x\to\pm\infty.$
 The corresponding Jost solutions
 $\psi(\zeta,x)$ and $\bar\psi(\zeta,x)$
 are explicitly expressed in 
 \eqref{6.23}
 with the help of \eqref{6.24}--\eqref{6.27}, where those expressions contain matrix
 exponentials. The equivalent expressions expressed in terms of elementary
 functions are extremely lengthy and hence 
 we do not display them here.
 Using \eqref{6.28} and \eqref{6.29}, we obtain 
 the transmission coefficients
 as
 \begin{equation*}
 T(\zeta)=\displaystyle\frac{(\lambda+i)^2(\lambda+2i)}{(\lambda-i)^2(\lambda-2i)},\quad \bar T(\zeta)=\displaystyle\frac{(\lambda-i)^2(\lambda-2i)}{(\lambda+i)^2(\lambda+2i)},
 \end{equation*}
 where we again recall that $\lambda$ and $\zeta$ are related to each other as in \eqref{1.13}.
 The corresponding quantity $E(x)$ is expressed as in \eqref{6.19},
 but the corresponding equivalent expression expressed in terms of elementary function is too lengthy 
 to displayed here .
 In this example, 
the constant $\mu$ in \eqref{3.23} is evaluated using
the second equality in \eqref{6.19} and we have $\mu=6\pi.$

\end{example}

\section{The conclusion}
\label{section7}

In this paper we present the solution to the inverse scattering problem for the linear system \eqref{1.1} by establishing the relevant Marchenko inversion method.
This is done by deriving the Marchenko system \eqref{5.31} of linear integral equations,
where the kernel and the nonhomogeneous term are expressed as in \eqref{5.29} in terms of the two reflection coefficients 
$R(\zeta)$ and $\bar R(\zeta)$ and the matrix triplet pair $(A,B,C)$ and $(\bar A,\bar B,\bar C)$ describing the bound-state information.
In \eqref{5.45} and \eqref{5.46} we show how the potentials $q(x)$ and $r(x),$
respectively are recovered from the solution to the Marchenko system \eqref{5.31}.
The representation of the bound-state information in terms of
a pair of matrix triplets allows us to deal with any number of bound states and with any multiplicity
for each bound state.
In the reflectionless case, the kernel and the nonhomogeneous terms in the Marchenko system \eqref{5.31} each
become separable. This yields explicit solutions to the Marchenko system expressed solely in terms of the matrix triplet pair,
and consequently we obtain the closed-form expressions for the potentials $q(x)$ and
$r(x)$ explicitly expressed in terms of elementary functions.

When we use the time-evolved scattering data set as input to the Marchenko system 
\eqref{5.31}, the recovered potentials $q$ and $r$ each become functions of the spacial
variable $x$ and the time variable $t.$ In that case, the time-evolved potential pair $(q,r)$ yields a
solution to the Gerdjikov--Ivanov system \eqref{1.5}. In a future publication we will elaborate on this issue and
we will also present explicit solutions to the integrable nonlinear system \eqref{1.5}.
Here, we only describe the time evolution of the scattering data set for \eqref{1.1}.

The time evolution of the scattering data set for the linear system \eqref{1.7} is known \cite{AEU2023b}.
Using the first equalities in \eqref{3.52}--\eqref{3.57} and the known time evolution of the scattering coefficients for 
\eqref{1.7}, we obtain the time evolution of the scattering coefficients for \eqref{1.1}
as
\begin{equation*}
T(\zeta,t)=T(\zeta,0),\quad 
\bar T(\zeta,t)=\bar T(\zeta,0),
\end{equation*}
\begin{equation}
\label{7.2}
R(\zeta,t)=R(\zeta,0)\, e^{4 i \zeta^4 t},\quad 
\bar R(\zeta,t)=\bar R(\zeta,0)\, e^{-4 i \zeta^4 t}, 
\end{equation}
\begin{equation*}
L(\zeta,t)=L(\zeta,0)\, e^{-4 i \zeta^4 t},\quad 
\bar L(\zeta,t)=\bar L(\zeta,0)\, e^{4 i \zeta^4 t}.
\end{equation*}
As for the time evolution of the matrix triplets $(A,B,C)$ and
$(\bar A,\bar B,\bar C)$ appearing in \eqref{4.2} and \eqref{4.6},
respectively, we mention that the matrices $A,$ $B,$ $\bar A,$ and $\bar B$
remain unchanged in time and that the matrices $C$ and $\bar C$ evolve in time as
\begin{equation}
\label{7.4}
C(t)=C(0)\,e^{4iA^2 t},\quad \bar C(t)=\bar C(0)\,e^{-4i\bar A^2 t}.
\end{equation}
Using \eqref{5.2}, \eqref{7.2}, and \eqref{7.4} in \eqref{5.29}, we see that the time-evolved
kernels for the Marchenko system \eqref{5.31} are given by
\begin{equation*}
\Omega(y,t)=\displaystyle\frac{1}{2\pi}\displaystyle\int_{-\infty}^\infty  
d\lambda\,\displaystyle\frac{R(\zeta)}{\zeta}\,e^{4i\lambda^2 t+i\lambda y}+C\,e^{4iA^2 t+iAy} B,
\end{equation*}
\begin{equation*}
\bar\Omega(y,t)=\displaystyle\frac{1}{2\pi}
\displaystyle\int_{-\infty}^\infty  d\lambda\,\displaystyle\frac{\bar R(\zeta)}{\zeta}\,e^{-4i\lambda^2 t-i\lambda y}+\bar C\,e^{-4i\bar A^2 t-i\bar Ay} \bar B,
\end{equation*}
where we recall that $\lambda$ is related to $\zeta$ as in \eqref{1.11}.

\end{document}